\documentclass[acmtocl]{acmtrans2m_calc}

\usepackage{amssymb}
\usepackage{amsmath}
\usepackage{stmaryrd}
\usepackage{gastex}

\newtheorem{theorem}{Theorem}[section]

\newtheorem{proposition}[theorem]{Proposition}

\newdef{example}[theorem]{Example}
\newdef{remark}[theorem]{Remark}

\newcommand{\tuple}[1]{\langle #1 \rangle}
\newcommand{\length}[1]{|#1|}

\newcommand{\ndownarrow}{\not\,\downarrow}

\newcommand{\edown}{\triangledown}
\newcommand{\eup}{\vartriangle}
\newcommand{\eright}{\vartriangleright}
\newcommand{\eleft}{\vartriangleleft}

\newcommand{\aalphabet}{\Sigma}
\newcommand{\aletter}{a}
\newcommand{\aletterbis}{b}
\newcommand{\emptyword}{\varepsilon}

\newcommand{\nodes}{N}

\newcommand{\anode}{n}

\newcommand{\atree}{t}
\newcommand{\adatatree}{\tau}

\newcommand{\locs}{Q}
\newcommand{\aloc}{q}
\newcommand{\locsbis}{R}
\newcommand{\alocbis}{r}
\newcommand{\locster}{S}

\newcommand{\atransf}{\varphi}

\newcommand{\aregaut}{\mathcal{A}}
\newcommand{\aregautbis}{\mathcal{B}}

\newcommand{\adomain}{\mathcal{D}}
\newcommand{\adatum}{D}
\newcommand{\adatumbis}{E}

\newcommand{\aconf}{G}
\newcommand{\aconfbis}{H}
\newcommand{\aconfs}{\mathcal{G}}
\newcommand{\aconfsbis}{\mathcal{H}}
\newcommand{\aconfss}{\mathbb{G}}
\newcommand{\aconfssbis}{\mathbb{H}}

\newcommand{\acaut}{\mathcal{C}}
\newcommand{\acval}{v}
\newcommand{\acvalbis}{w}
\newcommand{\ablock}{B}

\newcommand{\atreeaut}{\mathcal{T}}

\newcommand{\aquery}{p}
\newcommand{\aqual}{u}
\newcommand{\emptyquery}{\varepsilon}

\title
{Alternating Automata on Data Trees and \\
 XPath Satisfiability}

\author
{MARCIN JURDZI\'NSKI and RANKO LAZI\'C \\
 Department of Computer Science, University of Warwick, UK}

\begin{abstract}
A data tree is an unranked ordered tree
whose every node is labelled by
a letter from a finite alphabet and
an element (``datum'') from an infinite set,
where the latter can only be compared for equality.
The article considers alternating automata on data trees
that can move downward and rightward,
and have one register for storing data.
The main results are that nonemptiness over finite data trees
is decidable but not primitive recursive,
and that nonemptiness of safety automata
is decidable but not elementary.
The proofs use nondeterministic tree automata with faulty counters.
Allowing upward moves, leftward moves, or two registers,
each causes undecidability.
As corollaries, decidability is obtained for
two data-sensitive fragments of the XPath query language.
\end{abstract}





\begin{document}

\begin{bottomstuff}
This article is a revised and extended version of \cite{Jurdzinski&Lazic07}.
\newline
The second author was supported by a grant from the Intel Corporation.
\end{bottomstuff}

\maketitle

\section{Introduction}

\subsubsection*{Context}

Logics and automata for words and trees over finite alphabets
are relatively well-understood.
Motivated partly by the search for automated reasoning techniques for XML
and the need for formal verification and synthesis of infinite-state systems,
there is an active and broad research programme on
logics and automata for words and trees which have richer structure.

Initial progress made on reasoning about data words and data trees
is summarised in the survey by \citeN{Segoufin06}.
A data word is a word over $\aalphabet \times \adomain$,
where $\aalphabet$ is a finite alphabet,
and $\adomain$ is an infinite set (``domain'')
whose elements (``data'') can only be compared for equality.
Similarly, a data tree is a tree (countable, unranked and ordered)
whose every node is labelled by a pair in $\aalphabet \times \adomain$.

First-order logic for data words was considered by \citeN{Bojanczyketal06a},
and related automata were studied further by \citeN{Bjorklund&Schwentick07}.
The logic has variables which range over word positions
($\{0, \ldots, l - 1\}$ or $\mathbb{N}$),
a unary predicate for each letter from the finite alphabet,
and a binary predicate $x \sim y$ which denotes equality of data labels.
FO$^2(+1, <, \sim)$ denotes such a logic with two variables
and binary predicates $x + 1 = y$ and $x < y$.
Over finite and over infinite data words,
satisfiability for FO$^2(+1, <, \sim)$ was shown decidable
and at least as hard as nonemptiness of vector addition automata.
Whether the latter problem is elementary has been open for many years.
Extending the logic by one more variable causes undecidability.

Over data trees, FO$^2(+1, <, \sim)$ denotes a similar
first-order logic with two variables.
The variables range over tree nodes,
$+1$ stands for two predicates ``child'' and ``next sibling'',
and $<$ stands for two predicates ``descendant'' and ``younger sibling''.
Complexity of satisfiability over finite data trees
was studied by \citeN{Bojanczyketal09}.
For FO$^2(+1, \sim)$, it was shown to be in $3$\textsc{NExpTime},
but for FO$^2(+1, <, \sim)$, to be at least as hard as
nonemptiness of vector addition tree automata.
Decidability of the latter is an open question,
and it is equivalent to decidability of multiplicative exponential linear logic
\cite{deGroote&Guillaume&Salvati04}.
However, \citeN{Bjorklund&Bojanczyk07} showed that FO$^2(+1, <, \sim)$
over finite data trees of bounded depth is decidable.

XPath \cite{Clark&DeRose99} is a prominent query language
for XML documents \cite{Bray&Paoli&Sperberg-McQueen98}.
The most basic static analysis problem for XPath,
with a variety of applications, is satisfiability in the presence of DTDs.
In the two extensive articles on its complexity
\cite{Benedikt&Fan&Geerts08,Geerts&Fan05},
the only decidability result that allows negation and data
(i.e., equality comparisons between attribute values)
does not allow axes which are recursive
(such as ``self or descendant'') or between siblings.
By representing XML documents as data trees and
translating from XPath to FO$^2(+1, \sim)$, \citeN{Bojanczyketal09}
obtained a decidable fragment with negation, data and all nonrecursive axes.
Another fragment of XPath was considered
by \citeN{Halle&Villemaire&Cherkaoui06},
but it lacks concatenation, recursive axes and sibling axes.
A recent advance of \citeN{Figueira09} shows \textsc{ExpTime}-completeness
for full downward XPath, but with restricted DTDs.

An alternative approach to reasoning about data words
is based on automata with registers \cite{Kaminski&Francez94}.
A register is used for storing a datum for later equality comparisons.
Nonemptiness of one-way nondeterministic register automata
over finite data words has relatively low complexity:
\textsc{NP}-complete \cite{Sakamoto&Ikeda00} or
\textsc{PSpace}-complete \cite{Demri&Lazic09},
depending on technical details of their definition.
Unfortunately, such automata fail to provide a satisfactory notion
of regular language of finite data words,
as they are not closed under complement \cite{Kaminski&Francez94}
and their nonuniversality is undecidable \cite{Neven&Schwentick&Vianu04}.
To overcome those limitations,
one-way alternating automata with $1$ register
were proposed by \citeN{Demri&Lazic09}:
they are closed under Boolean operations,
their nonemptiness over finite data words is decidable,
and future-time fragments of temporal logics such as
LTL or the modal $\mu$-calculus extended by $1$ register
are easily translatable to such automata.
However, the nonemptiness problem over finite data words
turned out to be not primitive recursive.
Moreover, already with weak acceptance \cite{Muller&Saoudi&Schupp86}
and thus also with B\"uchi or co-B\"uchi acceptance,
nonemptiness over infinite data words is undecidable
(more precisely, co-r.e.-hard).
When the automata are restricted to those which recognise safety properties
\cite{Alpern&Schneider87} over infinite data words,
nonemptiness was shown to be \textsc{ExpSpace}-complete,
and inclusion to be decidable but not primitive recursive
\cite{Lazic06}.

\subsubsection*{Contribution}

This article addresses one of the research directions
proposed by \citeN{Segoufin06}: investigating
modal logics and automata with registers on data trees.
Nondeterministic automata with registers which can be
nondeterministically reassigned on finite binary data trees
were recently studied by \citeN{Kaminski&Tan08}:
top-down and bottom-up variants recognise the same languages,
and nonemptiness is decidable.
However, they inherit the drawbacks of
one-way nondeterministic register automata on data words:
lack of closure under complement and undecidability of nonuniversality.

We consider alternating automata that have $1$ register
and are forward, i.e., can move downward and rightward over tree nodes:
for short, ATRA$_1$.
They are closed under Boolean operations,
and we show that their nonemptiness over finite data trees is decidable.
Moreover, forward fragments of
CTL and the modal $\mu$-calculus with $1$ register
are easily translatable to ATRA$_1$ \cite{Jurdzinski&Lazic07}.
The expressiveness of ATRA$_1$ is incomparable to those of
FO$^2(+1, \sim)$ and the automata of \citeN{Kaminski&Tan08}:
for example, the latter two formalisms but not ATRA$_1$
can check whether some two leaves have equal data,
and the opposite is true of checking whether each node's datum is fresh,
i.e., does not appear at any ancestor node.
By lower-bound results for register automata on data words in
\cite{Neven&Schwentick&Vianu04,David04,Demri&Lazic09},
we have that ATRA$_1$ nonemptiness is not primitive recursive,
and that it becomes undecidable (more precisely, r.e.-hard)
if any of the following is added:
upward moves, leftward moves, or one more register.

Motivated partly by applications to XML streams
(cf., e.g., \cite{Olteanu&Furche&Bry04}),
we consider both finite and countably infinite data trees,
where horizontal as well as vertical infinity is allowed.
For ATRA$_1$ with the weak acceptance mechanism,
the undecidability result over infinite data words
\cite{Demri&Lazic09} carries over.
%By the \emph{finitary} variant of
%a decision problem over trees or data trees,
%we mean its restriction to finite structures.
However, we show that, for safety ATRA$_1$,
which are closed under intersection and union but not complement,
inclusion is decidable and not primitive recursive.
When a data tree is rejected by
an automaton with the safety acceptance mechanism,
there exists an initial segment whose every extension is rejected.
We also obtain that nonemptiness of safety ATRA$_1$ is not elementary.
The latter is the most surprising result in the article:
it means that the techniques in the proof that
nonemptiness over infinite data words of
safety one-way alternating automata with $1$ register
is in \textsc{ExpSpace} cannot be lifted to trees
to obtain a $2$\textsc{ExpTime} upper bound.

The proofs of decidability
involve translating from ATRA$_1$ to
forward nondeterministic tree automata with faulty counters.
The counters are faulty in the sense that
they are subject to incrementing errors,
i.e., can spontaneously increase at any time.
That makes the transition relations downwards compatible
with a well-quasi-ordering (cf.\ \cite{Finkel&Schnoebelen01}),
which leads to lower complexities of some verification problems
than with error-free counters.

We define forward XPath to be the largest downward and rightward fragment
in which, whenever two attribute values are compared for equality,
one of them must be at the current node.
By translating from forward XPath to ATRA$_1$,
we obtain decidability of satisfiability over finite documents
and decidability of satisfiability for a safety subfragment,
both in the presence of DTDs.
%The complexity of finitary satisfiability is not primitive recursive,
%even without sibling axes.
In contrast to the decidable fragments of XPath mentioned previously,
forward XPath has sibling axes, recursive axes, concatenation,
negation, and data comparisons.

\section{Preliminaries}

After fixing notations for trees and data trees,
we define two kinds of forward automata
and look at some of their basic properties:
alternating automata with $1$ register on data trees, and
nondeterministic automata with counters with incrementing errors on trees.

\subsection{Trees and Data Trees}
\label{ss:trees}

For technical simplicity, we shall work with binary trees
instead of unranked ordered trees.
Firstly, as e.g.\ \citeN{Bjorklund&Bojanczyk07},
we adopt the insignificant generalisation of
considering unranked ordered forests,
in which the roots are regarded as siblings with no parent.
Secondly, the following is a standard and trivial
one-to-one correspondence between
unranked ordered forests $\atree$ and binary trees $\mathrm{bt}(\atree)$:
the nodes of $\mathrm{bt}(\atree)$ are the same as the nodes of $\atree$,
and the children of each node $\anode$ in $\mathrm{bt}(\atree)$
are the first child and next sibling of $\anode$ in $\atree$.
The correspondence works for finite as well as
infinite unranked ordered forests.
In the latter, there may be infinite (of type $\omega$)
branches or siblinghoods or both.

Without loss of generality,
each node will either have both children or be a leaf,
only nonleaf nodes will be labelled,
and the root node will be nonleaf.
Formally, a \emph{tree} is a tuple
$\tuple{\nodes, \aalphabet, \Lambda}$, where:
\begin{itemize}
\item
$\nodes$ is a prefix-closed subset of $\{0, 1\}^*$
such that $\length{\nodes} > 1$ and,
for each $\anode \in \nodes$,
either $\anode \cdot 0 \in \nodes$
and $\anode \cdot 1 \in \nodes$,
or $\anode \cdot 0 \notin \nodes$
and $\anode \cdot 1 \notin \nodes$;
\item
$\aalphabet$ is a finite alphabet;
\item
$\Lambda$ is a mapping from the nonleaf elements of $\nodes$ to $\aalphabet$.
\end{itemize}

A \emph{data tree} is a tree as above together with a mapping
$\Delta$ from the nonleaf nodes to a fixed infinite set $\adomain$.
For a data tree $\adatatree$, let $\mathrm{tree}(\adatatree)$
denote the underlying tree.

For a data tree $\adatatree$ and $l > 0$,
let the \emph{$l$-prefix} of $\adatatree$ be the data tree
obtained by restricting $\adatatree$ to nodes of length at most $l$.
For each $\aalphabet$, the set of all data trees with alphabet $\aalphabet$
is a complete metric space with the following notion of distance:
for distinct $\adatatree$ and $\adatatree'$,
let $d(\adatatree, \adatatree') = 1 / l$ where $l$ is least such that
$\adatatree$ and $\adatatree'$ have distinct $l$-prefixes.

\subsection{Alternating Tree Register Automata}
\label{ss:ATRA}

\subsubsection*{Automata}

A run of a forward alternating automaton with $1$ register
on a data tree will consist of a \emph{configuration} for each tree node.
Each configuration will be a finite set of \emph{threads},
which are pairs of an automaton state and a register value,
where the latter is a datum from $\adomain$.

Following \citeN{Brzozowski&Leiss80},
transitions will be specified by positive Boolean formulae.
For a set of states $\locs$, let $\mathcal{B}^+(\locs)$ consist of
all formulae given by the following grammar, where $\aloc \in \locs$:
\[\atransf \:::=\:
  \aloc(0, \downarrow) \,\mid\, \aloc(0, \ndownarrow) \,\mid\,
  \aloc(1, \downarrow) \,\mid\, \aloc(1, \ndownarrow) \,\mid\,
  \top \,\mid\, \bot \,\mid\,
  \atransf \wedge \atransf \,\mid\, \atransf \vee \atransf\]

Given a configuration $\aconf$ at a nonleaf tree node $\anode$,
for each thread $\tuple{\aloc, \adatum}$ in $\aconf$,
the automaton transition function provides
a formula $\atransf$ in $\mathcal{B}^+(\locs)$,
which depends on $\aloc$, on the letter labelling $\anode$,
and on whether $\adatum = \adatumbis$,
where $\adatumbis$ is the datum labelling $\anode$.
In $\atransf$, an atom $\alocbis(d, \downarrow)$
requires that thread $\tuple{\alocbis, \adatumbis}$
be in the configuration for node $\anode \cdot d$
(i.e., the register value is replaced by the datum at $\anode$),
and an atom $\alocbis(d, \ndownarrow)$ requires the same
for thread $\tuple{\alocbis, \adatum}$
(i.e., the register value is not replaced).

Formally, a \emph{forward alternating tree $1$-register automaton}
(shortly, \emph{ATRA$_1$}) $\aregaut$ is a tuple
$\tuple{\aalphabet, \locs, \aloc_I, F, \delta}$ such that:
\begin{itemize}
\item
$\aalphabet$ is a finite alphabet and
$\locs$ is a finite set of states;
\item
$\aloc_I \in \locs$ is the initial state and
$F \subseteq \locs$ are the final states;
\item
$\delta:
 \locs \times \aalphabet \times \{\mathit{tt}, \mathit{ff}\} \,\rightarrow\,
 \mathcal{B}^+(\locs)$
is a transition function.
\end{itemize}

\subsubsection*{Runs and Languages}

The semantics of the positive Boolean formulae can be given
by defining when a quadruple
$\locsbis_0^\downarrow, \locsbis_0^{\ndownarrow},
 \locsbis_1^\downarrow, \locsbis_1^{\ndownarrow}$
of subsets of $\locs$ satisfies a formula $\atransf$ in $\mathcal{B}^+(\locs)$,
by structural recursion.
The cases for the Boolean atoms and operators are standard,
and for the remaining atoms we have:
\[\locsbis_0^\downarrow, \locsbis_0^{\ndownarrow},
  \locsbis_1^\downarrow, \locsbis_1^{\ndownarrow}
  \models
  \alocbis(d, ?)
  \,\stackrel{\mathrm{def}}{\Leftrightarrow}\,
  \alocbis \in \locsbis_d^?\]

We can now define the transition relation of $\aregaut$,
which is between configurations and pairs of configurations,
and relative to a letter and a datum.
We write 
$\aconf
 \rightarrow_\aletter^\adatumbis
 \aconfbis_0, \aconfbis_1$
iff, for each thread $\tuple{\aloc, \adatum} \in \aconf$, there exist
$\locsbis_0^\downarrow, \locsbis_0^{\ndownarrow},
 \locsbis_1^\downarrow, \locsbis_1^{\ndownarrow}
 \models
 \delta(\aloc, \aletter, \adatum = \adatumbis)$
such that, for both $d \in \{0, 1\}$:
\[\{\tuple{\alocbis, \adatumbis} \,:\,
    \alocbis \in \locsbis_d^\downarrow\} \,\cup\,
  \{\tuple{\alocbis, \adatum} \,:\,
    \alocbis \in \locsbis_d^{\ndownarrow}\}
  \:\subseteq\: \aconfbis_d\]

A run of $\aregaut$ on a data tree
$\tuple{\nodes, \aalphabet, \Lambda, \Delta}$
is a mapping $\anode \mapsto \aconf_\anode$
from the nodes to configurations such that:
\begin{itemize}
\item
the initial thread is in the configuration at the root, i.e.\
$\tuple{\aloc_I, \Delta(\emptyword)} \in \aconf_\emptyword$;
\item
for each nonleaf $\anode$, the transition relation is observed, i.e.\
$\aconf_\anode
 \rightarrow_{\Lambda(\anode)}^{\Delta(\anode)}
 \aconf_{\anode \cdot 0}, \aconf_{\anode \cdot 1}$.
\end{itemize}
We say that the run is:
\begin{itemize}
\item
\emph{final} iff, for each leaf $\anode$,
only final states occur in $\aconf_\anode$;
\item
\emph{finite} iff there exists $l$ such that,
for each $\anode$ of length at least $l$,
$\aconf_\anode$ is empty.
\end{itemize}

We may regard $\aregaut$ as
an automaton on finite data trees,
a safety automaton, or
a co-safety automaton.
We say that:
\begin{itemize}
\item
$\aregaut$ \emph{accepts} a finite data tree $\adatatree$
iff $\aregaut$ has a final run on $\adatatree$;
\item
$\aregaut$ \emph{safety-accepts} a data tree $\adatatree$
iff $\aregaut$ has a final run on $\adatatree$;
\item
$\aregaut$ \emph{co-safety-accepts} a data tree $\adatatree$
iff $\aregaut$ has a final finite run on $\adatatree$.
\end{itemize}
Observe that, for finite data trees,
the three modes of $\aregaut$ coincide.

Let $\mathrm{L}^{\mathit fin}(\aregaut)$ denote
the set of all finite data trees with alphabet $\aalphabet$
that $\aregaut$ accepts,
and $\mathrm{L}^{\mathit saf}(\aregaut)$
(resp., $\mathrm{L}^{\mathit cos}(\aregaut)$) denote
the set of all data trees with alphabet $\aalphabet$
that $\aregaut$ safety-accepts (resp., co-safety-accepts).

\begin{remark}
\label{r:junk}
The valid initial and successor configurations in runs
were defined in terms of lower bounds on sets.
In other words, while running on any data tree, at each node
the automaton is free to introduce arbitrary ``junk'' threads.
However, final and finite runs were defined in terms of upper bounds on sets,
so junk threads can only make it harder
to complete a partial run into an accepting one.
This will play an important role in the proof of
decidability in Theorem~\ref{th:ATRA1}.
\end{remark}

\subsubsection*{Boolean Operations}

Given an ATRA$_1$ $\aregaut$, let $\overline{\aregaut}$ denote its dual:
the automaton obtained by replacing the set of final states with its complement
and replacing, in each transition formula $\delta(\aloc, \aletter, p)$,
every $\top$ with $\bot$, every $\wedge$ with $\vee$, and vice versa.
Observe that $\overline{\overline{\aregaut}} = \aregaut$.
Considering $\aregaut$ (resp., $\overline{\aregaut}$) as
a weak alternating automaton whose every state is of even (resp., odd) parity,
we have by \cite[Theorem~1]{Loding&Thomas00} that
$\mathrm{L}^{\mathit cos}(\overline{\aregaut})$ is the complement of
$\mathrm{L}^{\mathit saf}(\aregaut)$.
Hence, we also have that
$\mathrm{L}^{\mathit saf}(\overline{\aregaut})$ is the complement of
$\mathrm{L}^{\mathit cos}(\aregaut)$,
and that
$\mathrm{L}^{\mathit fin}(\overline{\aregaut})$ is the complement of
$\mathrm{L}^{\mathit fin}(\aregaut)$.

For each $\mathrm{m}$ of $\mathit{fin}$, $\mathit{saf}$, $\mathit{cos}$,
given ATRA$_1$ $\aregaut_1$ and $\aregaut_2$ with alphabet $\aalphabet$,
an automaton whose language in mode $\mathrm{m}$ is
$\mathrm{L}^{\mathrm m}(\aregaut_1) \cap
 \mathrm{L}^{\mathrm m}(\aregaut_2)$
(resp.,
$\mathrm{L}^{\mathrm m}(\aregaut_1) \cup
 \mathrm{L}^{\mathrm m}(\aregaut_2)$)
is constructible easily.
It suffices to form a disjoint union of $\aregaut_1$ and $\aregaut_2$,
and add a new initial state $\aloc_I$ such that
$\delta(\aloc_I, \aletter, \mathit{tt}) =
 \delta(\aloc_I^1, \aletter, \mathit{tt}) \wedge
 \delta(\aloc_I^2, \aletter, \mathit{tt})$
(resp.,
$\delta(\aloc_I, \aletter, \mathit{tt}) =
 \delta(\aloc_I^1, \aletter, \mathit{tt}) \vee
 \delta(\aloc_I^2, \aletter, \mathit{tt})$)
for each $\aletter \in \aalphabet$,
where $\aloc_I^1$ and $\aloc_I^2$ are the initial states of
$\aregaut_1$ and $\aregaut_2$.
(Since the initial thread's register value is always the root node's datum,
the formulae $\delta(\aloc_I, \aletter, \mathit{ff})$ are irrelevant.)

We therefore obtain:

\begin{proposition}
\label{pr:closure.RA}
\begin{itemize}
\item[(a)]
ATRA$_1$ on finite data trees are closed under
complement, intersection and union.
\item[(b)]
Safety ATRA$_1$ and co-safety ATRA$_1$ are dual,
and each is closed under intersection and union.
\end{itemize}
In each case, a required automaton is computable in logarithmic space.
\end{proposition}

\subsubsection*{Safety Languages}

A set $L$ of data trees with alphabet $\aalphabet$
is called \emph{safety} \cite{Alpern&Schneider87} iff
it is closed with respect to the metric defined in Section~\ref{ss:trees},
i.e.\ for each data tree $\adatatree$,
if for all $l > 0$ there exists $\adatatree'_l \in L$
such that the $l$-prefixes of $\adatatree$ and $\adatatree'_l$ are equal,
then $\adatatree \in L$.
The complements of safety languages, i.e.\ the open sets of data trees,
are called \emph{co-safety}.

\begin{proposition}
\label{pr:safety.RA}
For each ATRA$_1$ $\aregaut$, we have that
$\mathrm{L}^{\mathit saf}(\aregaut)$ is safety and
$\mathrm{L}^{\mathit cos}(\aregaut)$ is co-safety.
\end{proposition}

\begin{proof}
By Proposition~\ref{pr:closure.RA}(b), it suffices to show that
$\mathrm{L}^{\mathit saf}(\aregaut)$ is safety.
Suppose for all $l > 0$
there exists $\adatatree'_l \in \mathrm{L}^{\mathit saf}(\aregaut)$
such that the $l$-prefixes of $\adatatree$ and $\adatatree'_l$ are equal.

For each $l > 0$, let us fix a final run
$\anode \mapsto \aconf'_{l, \anode}$
of $\aregaut$ on $\adatatree'_l$.
For each $0 \leq k \leq l$, let $\aconfs_{l, k}$ denote
the restriction of the run $\anode \mapsto \aconf'_{l, \anode}$
to nodes $\anode$ of length $k$.

Consider the tree consisting of the empty sequence and all sequences
$\aconfs_{l, 0} \cdot \aconfs_{l, 1} \cdot \cdots \aconfs_{l, k}$
for $l > 0$ and $0 \leq k \leq l$.
Without loss of generality, each register value in each $\aconf'_{l, \anode}$
labels some node of $\adatatree'_l$ on the path from the root to $\anode$,
so the tree is finitely branching.
By K\"onig's Lemma, it has an infinite path
$\aconfsbis_0 \cdot \aconfsbis_1 \cdot \cdots$.
For each $0 \leq k$, $\aconfsbis_k$ is a mapping from
the nodes of $\adatatree$ of length $k$ to configurations of $\aregaut$.
It remains to observe that
$\anode \mapsto \aconfsbis_{\length{\anode}}(\anode)$
is a final run of $\aregaut$ on $\adatatree$.
\end{proof}

\begin{example}
\label{ex:ATRA}
By recursion on $k \geq 1$, we shall define ATRA$_1$ $\aregautbis_k$
with alphabet $\{\aletterbis_1, \ldots, \aletterbis_k, *\}$.
As well as being interesting examples of ATRA$_1$,
the $\aregautbis_k$ will be used in the nonelementarity part of
the proof of Theorem~\ref{th:safety}.

Let $\aregautbis_1$ be the automaton depicted in Figure~\ref{f:B1}.
It has three states, where $\aloc$ is initial, and $\aloc''$ is final.
We have
$\delta(\aloc, \aletterbis_1, p) =
 \aloc'(0, \ndownarrow) \wedge
 \aloc''(1, \ndownarrow)$
and
$\delta(\aloc', \aletterbis_1, p) =
 \aloc''(0, \ndownarrow) \wedge
 \aloc''(1, \ndownarrow)$
for both $p \in \{\mathit{tt}, \mathit{ff}\}$,
and the transition function gives $\bot$ in all other cases.
(Recalling that the initial thread's register value is the root node's datum,
the formula $\delta(\aloc, \aletterbis_1, \mathit{ff})$ is in fact irrelevant.)
Observe that $\aregautbis_1$ safety-accepts exactly data trees
that have two nonleaf nodes, the root and its left-hand child,
and both are labelled by letter $\aletterbis_1$.

\begin{narrowfig}{.67\textwidth}
\setlength{\unitlength}{2em}
\begin{picture}(14,3.5)(-.5,-.5)
\gasset{Nadjust=wh}
\node[Nmarks=i,iangle=180,ilength=1](1)(1,1){$\aloc$}
\node[Nadjustdist=0](11)(4,1){}
\node(2)(7,2.5){$\aloc'$}
\node[Nadjustdist=0](21)(10,2.5){}
\node[Nmarks=r](3)(13,1){$\aloc''$}
\drawedge[AHnb=0](1,11){$\aletterbis_1$}
\drawedge(11,2){$0, \ndownarrow$}
\drawedge[AHnb=0](2,21){$\aletterbis_1$}
\drawedge[curvedepth=.5](21,3){$0, \ndownarrow$}
\drawedge[curvedepth=-.5,ELside=r,ELpos=25](21,3){$1, \ndownarrow$}
\drawedge[curvedepth=-1,ELside=r](11,3){$1, \ndownarrow$}
\end{picture}
\caption{Defining $\aregautbis_1$}
\label{f:B1}
\end{narrowfig}

For each $k \geq 1$, $\aregautbis_{k + 1}$ is defined so that it safety-accepts
a data tree over $\{\aletterbis_1, \ldots, \aletterbis_{k + 1}, *\}$ iff:
\begin{itemize}
\item[(i)]
the root node is labelled by $\aletterbis_{k + 1}$,
its left-hand child is labelled by $\aletterbis_{k + 1}$,
and its right-hand child is a leaf;
\item[(ii)]
for each node $\anode$ labelled by $\aletterbis_{k + 1}$, which is not the root,
the left-hand child of $\anode$ is labelled by $*$
and its both children are labelled by $\aletterbis_{k + 1}$,
and the right-hand subtree at $\anode$ is safety-accepted by $\aregautbis_k$;
\item[(iii)]
whenever a node $\anode$, which is not the root,
and a descendant $\anode'$ of $\anode$
are labelled by $\aletterbis_{k + 1}$,
we have that their data labels are distinct,
and that the datum at $\anode$ equals the datum at some node
which is labelled by $\aletterbis_k$ and
which is in the right-hand subtree at $\anode'$.
\end{itemize}
By Proposition~\ref{pr:closure.RA}(b),
it suffices to define automata for (i)--(iii) separately.
Expressing (i) and (ii) is straightforward,
and an automaton for (iii) is depicted in Figure~\ref{f:Bkp1}.
It has four states, where $\aloc_0$ is initial,
and $\aloc_1$ and $\aloc_2$ are final.
For all letters $\aletter$ and Booleans $p$, we have
$\delta(\aloc_0, \aletter, p) =
 \aloc_1(0, \ndownarrow)$,
so initially the automaton moves to the left-hand child of the root
and changes the state to $\aloc_1$.
From $\aloc_1$, if the current node is labelled by $*$,
the automaton moves to both children:
$\delta(\aloc_1, *, p) =
 \aloc_1(0, \ndownarrow) \wedge \aloc_1(1, \ndownarrow)$
for both $p$.
Also from $\aloc_1$, if the current node $\anode$
is labelled by $\aletterbis_{k + 1}$,
the automaton both moves to the left-hand child without changing the state,
and moves to the left-hand child with
storing the datum at $\anode$ in the register and
changing the state to $\aloc_2$:
$\delta(\aloc_1, \aletterbis_{k + 1}, p) =
 \aloc_1(0, \ndownarrow) \wedge \aloc_2(0, \downarrow)$
for both $p$.
From $\aloc_2$, the behaviour for $*$ is analogous to that from $\aloc_1$,
but if the current node's letter is $\aletterbis_{k + 1}$
and its datum is distinct from the datum in the register,
the automaton both moves to the left-hand child without changing the state
and moves to the right-hand child with changing the state to $\aloc_3$:
$\delta(\aloc_2, \aletterbis_{k + 1}, \mathit{ff}) =
 \aloc_2(0, \ndownarrow) \wedge \aloc_3(1, \ndownarrow)$.
The remainder of Figure~\ref{f:Bkp1} is interpreted similarly,
and in cases not depicted, the transition function gives $\bot$.
Since the mode of acceptance is safety, the automaton in fact expresses:
\begin{itemize}
\item[(iii')]
whenever a node $\anode$, which is not the root,
and a descendant $\anode'$ of $\anode$
are labelled by $\aletterbis_{k + 1}$,
we have that their data labels are distinct,
and that either the datum at $\anode$ equals the datum at some node
which is labelled by $\aletterbis_k$ and
which is in the right-hand subtree at $\anode'$,
or that subtree is infinite.
\end{itemize}

\begin{figure}
\setlength{\unitlength}{2em}
\begin{center}
\begin{picture}(20,6)(-.5,1)
\gasset{Nadjust=wh}
\node[Nmarks=i,iangle=180,ilength=1](0)(1,4){$\aloc_0$}
\node[Nmarks=r](1)(4,4){$\aloc_1$}
\drawedge[ELside=r](0,1){$0, \ndownarrow$}
\node[Nadjustdist=0](1a)(4,7){}
\drawedge[AHnb=0,ELside=r](1,1a){$*$}
\drawedge[curvedepth=-1,ELside=r](1a,1){$0, \ndownarrow$}
\drawedge[curvedepth=1](1a,1){$1, \ndownarrow$}
\node[Nadjustdist=0](1b)(7,4){}
\node[Nmarks=r](2)(10,4){$\aloc_2$}
\drawedge[AHnb=0,ELside=r](1,1b){$\aletterbis_{k + 1}$}
\drawedge[curvedepth=2](1b,1){$0, \ndownarrow$}
\drawedge[ELside=r](1b,2){$0, \downarrow$}
\node[Nadjustdist=0](2a)(10,7){}
\drawedge[AHnb=0,ELside=r](2,2a){$*$}
\drawedge[curvedepth=-1,ELside=r](2a,2){$0, \ndownarrow$}
\drawedge[curvedepth=1](2a,2){$1, \ndownarrow$}
\node[Nadjustdist=0](2b)(13,4){}
\node(3)(16,4){$\aloc_3$}
\drawedge[AHnb=0,ELside=r](2,2b){$\aletterbis_{k + 1}, \neq$}
\drawedge[curvedepth=2](2b,2){$0, \ndownarrow$}
\drawedge[ELside=r](2b,3){$1, \ndownarrow$}
\node[Nadjustdist=0](3a)(13,7){}
\drawedge[AHnb=0,curvedepth=.5](3,3a){$*$}
\drawedge[curvedepth=.5](3a,3){$0, \ndownarrow$}
\node[Nadjustdist=0](3b)(19,7){}
\drawedge[AHnb=0,curvedepth=.5](3,3b){$*$}
\drawedge[curvedepth=.5](3b,3){$1, \ndownarrow$}
\node[Nframe=n](3c)(19,4){$\top$}
\drawedge[ELside=r](3,3c){$\aletterbis_k, =$}
\node[Nadjustdist=0](3d)(16,1){}
\drawedge[AHnb=0,curvedepth=-.5,ELside=r](3,3d){$\aletterbis_k, \neq$}
\drawedge[curvedepth=-.5,ELside=r](3d,3){$0, \ndownarrow$}
\end{picture}
\end{center}
\caption{Defining $\aregautbis_{k + 1}$}
\label{f:Bkp1}
\end{figure}

Let $2 \Uparrow 0 = 1$,
and $2 \Uparrow k = 2^{2 \Uparrow (k - 1)}$ for $k \geq 1$.
By induction on $k \geq 1$, the safety language of $\aregautbis_k$
has the following two properties.
In particular, in the presence of (i) and (ii),
we have that (iii) and (iii') are equivalent.
\begin{itemize}
\item
for every $\adatatree$ safety-accepted by $\aregautbis_k$,
every downward sequence which is from the left-hand child of the root
and which consists of nodes labelled by $\aletterbis_k$
is of length at most $2 \Uparrow (k - 1)$,
so $\adatatree$ is finite and has at most
$2 \Uparrow k$ nodes labelled by $\aletterbis_k$;
\item
for some $\adatatree$ safety-accepted by $\aregautbis_k$,
the nodes labelled by $\aletterbis_k$ other than the root
form a full binary tree of height $2 \Uparrow (k - 1)$
(after removing the nodes labelled by $*$),
so there are $2 \Uparrow k$ nodes labelled by $\aletterbis_k$,
and moreover the data at those nodes are mutually distinct.
\end{itemize}

Finally, we observe that for computing $\aregautbis_k$,
space logarithmic in $k$ suffices.
\end{example}

\subsection{Faulty Tree Counter Automata}
\label{ss:ITCA}

In Section~\ref{s:dec.fin}, we shall establish decidability of nonemptiness of
forward alternating tree $1$-register automata over finite data trees,
by translating them to automata which have natural-valued counters
with increments, decrements and zero-tests.
The translation will eliminate conjunctive branchings, by having
configurations of the former automata (which are finite sets of threads)
correspond to pairs of states and counter valuations,
so the latter automata will be only nondeterministic.
Also, data will be abstracted in the translation,
so the counter automata will run on finite trees (without data).

The feature that will make nonemptiness of the counter automata decidable
(on finite trees) is that they will be faulty, in the sense that
one or more counters can erroneously increase at any time.
The key insight is that such faults do not affect the translation's
preservation of nonemptiness: they in fact correspond to introductions
of ``junk'' threads in runs of ATRA$_1$ (cf.\ Remark~\ref{r:junk}).

For clarity of the correspondence between the finitary languages of ATRA$_1$
and the languages of their translations,
the counter automata will have $\emptyword$-transitions.

We now define the counter automata,
and show their nonemptiness decidable.

\subsubsection*{Automata}

An \emph{incrementing tree counter automaton} (shortly, \emph{ITCA}) $\acaut$,
which is forward and with $\emptyword$-transitions,
is a tuple $\tuple{\aalphabet, \locs, \aloc_I, F, k, \delta}$ such that:
\begin{itemize}
\item
$\aalphabet$ is a finite alphabet and
$\locs$ is a finite set of states;
\item
$\aloc_I \in \locs$ is the initial state and
$F \subseteq \locs$ are the final states;
\item
$k \in \mathbb{N}$ is the number of counters;
\item
$\delta \,\subseteq\,
 (\locs \times \aalphabet \times L \times \locs \times \locs) \cup
 (\locs \times \{\emptyword\} \times L \times \locs)$
is a transition relation, where
$L = \{\mathtt{inc, dec, ifz}\} \times \{1, \ldots, k\}$
is the instruction set.
\end{itemize}

\subsubsection*{Runs and Languages}

A counter valuation is a mapping from $\{1, \ldots, k\}$ to $\mathbb{N}$.
For counter valuations $\acval$ and $\acval'$, we write:
\begin{center}
\begin{tabular}{rcl}
$\acval \leq \acval'$
& iff &
$\acval(c) \leq \acval'(c)$ for all $c$
\\
$\acval \xrightarrow{\tuple{\mathtt{inc}, c}}_\surd \acval'$
& iff &
$\acval' = \acval[c \mapsto \acval(c) + 1]$
\\
$\acval \xrightarrow{\tuple{\mathtt{dec}, c}}_\surd \acval'$
& iff &
$\acval' = \acval[c \mapsto \acval(c) - 1]$
\\
$\acval \xrightarrow{\tuple{\mathtt{ifz}, c}}_\surd \acval'$
& iff &
$\acval(c) = 0$ and $\acval' = \acval$
\\
$\acval \stackrel{l}{\rightarrow} \acval'$
& iff &
$\acval \leq \acval_\surd
 \stackrel{l}{\rightarrow}_\surd
 \acval'_\surd \leq \acval'$
for some $\acval_\surd$, $\acval'_\surd$
\end{tabular}
\end{center}

A \emph{configuration} of $\acaut$ is a pair $\tuple{\aloc, \acval}$,
where $\aloc$ is a state and $\acval$ is a counter valuation.

To define runs, we first specify that
a \emph{block} is a nonempty finite sequence of configurations
obtainable by performing $\emptyword$-transitions,
i.e.\ for every two adjacent configurations
$\tuple{\aloc_i, \acval_i}$ and $\tuple{\aloc_{i + 1}, \acval_{i + 1}}$
in a block, there exists $l$ with
$\tuple{\aloc_i, \emptyword, l, \aloc_{i + 1}} \in \delta$ and
$\acval_i \stackrel{l}{\rightarrow} \acval_{i + 1}$.

Now, a run of $\acaut$ on a finite tree $\tuple{\nodes, \aalphabet, \Lambda}$
is a mapping $\anode \mapsto \ablock_\anode$ from the nodes to blocks such that:
\begin{itemize}
\item
$\tuple{\aloc_I, \mathbf{0}}$ is the first configuration in
$\ablock_\emptyword$;
\item
for each nonleaf $\anode$, there exists $l$ with
$\tuple{\aloc, \Lambda(\anode), l, \alocbis_0, \alocbis_1} \in \delta$,
$\acval \stackrel{l}{\rightarrow} \acvalbis_0$ and
$\acval \stackrel{l}{\rightarrow} \acvalbis_1$,
where $\tuple{\aloc, \acval}$ is the last configuration in $\ablock_\anode$,
and $\tuple{\alocbis_0, \acvalbis_0}$ and $\tuple{\alocbis_1, \acvalbis_1}$
are the first configurations in
$\ablock_{\anode \cdot 0}$ and $\ablock_{\anode \cdot 1}$ (respectively).
\end{itemize}

We regard such a run accepting iff, for each leaf $\anode$,
the state of the last configuration in $\ablock_\anode$ is final.
The language $\mathrm{L}(\acaut)$ is the set of all finite trees
with alphabet $\aalphabet$ on which $\acaut$ has an accepting run.

\subsubsection*{Decidability of Nonemptiness}

We remark that, since nonemptiness of incrementing counter automata over words
is not primitive recursive \cite[Theorem~2.9(b)]{Demri&Lazic09},
the same is true of nonemptiness of ITCA.

\begin{theorem}
\label{th:ITCA}
Nonemptiness of ITCA is decidable.
\end{theorem}

\begin{proof}
Consider an ITCA
$\acaut = \tuple{\aalphabet, \locs, \aloc_I, F, k, \delta}$.

For counter valuations $\acval$ and $\acval'$, and an instruction $l$,
we say that $\acval$ under $l$ yields $\acval'$ \emph{lazily}
and write $\acval \stackrel{l}{\rightarrow}_\flat \acval'$ iff
either $\acval \stackrel{l}{\rightarrow}_\surd \acval'$
(i.e., there are no incrementing errors),
or $l$ is of the form $\tuple{\mathtt{dec}, c}$,
$\acval(c) = 0$ and $\acval' = \acval$
(i.e., $0$ is erroneously decremented to $0$).
Observe that:
\begin{describe}{(*)}
\item[(*)]
Whenever $\acval \leq \acvalbis$
and $\acvalbis \stackrel{l}{\rightarrow} \acvalbis'$,
there exists $\acval'$
such that $\acval \stackrel{l}{\rightarrow}_\flat \acval'$
and $\acval' \leq \acvalbis'$.
\end{describe}

To reduce the nonemptiness problem for $\acaut$ to a reachability problem,
let a \emph{level} of $\acaut$ be a finite set of configurations.
For levels $\aconfs$ and $\aconfs'$ of $\acaut$,
let us write $\aconfs \rightarrow \aconfs'$ iff
$\aconfs'$ can be obtained from $\aconfs$ as follows:
\begin{itemize}
\item
each $\tuple{\aloc, \acval} \in \aconfs$ with $\aloc \notin F$ is replaced
either by the two configurations that some firable transition
$\tuple{\aloc, \aletter, l, \alocbis_0, \alocbis_1}$ yields lazily,
or by the one configuration that some firable transition
$\tuple{\aloc, \emptyword, l, \alocbis}$ yields lazily;
\item
each $\tuple{\aloc, \acval} \in \aconfs$ with $\aloc \in F$
is removed.
\end{itemize}
Performing transitions of $\acaut$ lazily ensures that,
for every level $\aconfs$, the set
$\{\aconfs' \,:\, \aconfs \rightarrow \aconfs'\}$
of all its successors is finite.
The latter set is also computable.
By the definition of accepting runs and (*),
we have that $\acaut$ is nonempty iff
the empty level is reachable from
the initial level $\{\tuple{\aloc_I, \mathbf{0}}\}$.

For configurations $\tuple{\aloc, \acval}$ and $\tuple{\alocbis, \acvalbis}$,
let $\tuple{\aloc, \acval} \leq \tuple{\alocbis, \acvalbis}$ iff
$\aloc = \alocbis$ and $\acval \leq \acvalbis$.
Now, let $\preceq$ be the quasi-ordering obtained by
lifting $\leq$ to levels:
$\aconfs \preceq \aconfsbis$ iff,
for each $\tuple{\aloc, \acval} \in \aconfs$,
there exists $\tuple{\alocbis, \acvalbis} \in \aconfsbis$
such that $\tuple{\aloc, \acval} \leq \tuple{\alocbis, \acvalbis}$.
By Higman's Lemma \cite{Higman52}, $\preceq$ is a well-quasi-ordering,
i.e., for every infinite sequence $\aconfs_0, \aconfs_1, \ldots$,
there exist $i < j$ such that $\aconfs_i \preceq \aconfs_j$.
Observe that, in the terminology of \citeN{Finkel&Schnoebelen01},
$\preceq$ is strongly downward-compatible with $\rightarrow$:
whenever $\aconfs \preceq \aconfsbis$
and $\aconfsbis \rightarrow \aconfsbis'$,
there exists $\aconfs'$
such that $\aconfs \rightarrow \aconfs'$
and $\aconfs' \preceq \aconfsbis'$.
Also, $\preceq$ is decidable.

Since $\aconfs \preceq \emptyset$ iff $\aconfs = \emptyset$,
we have reduced nonemptiness of $\acaut$ to
the subcovering problem for downward well-structured transition systems
with reflexive (which is weaker than strong) compatibility,
computable successor sets and decidable ordering.
The latter is decidable by \cite[Theorem~5.5]{Finkel&Schnoebelen01}.
\end{proof}

\section{Decidability Over Finite Data Trees}
\label{s:dec.fin}

\begin{theorem}
\label{th:ATRA1}
Nonemptiness of ATRA$_1$ over finite data trees is
decidable and not primitive recursive.
\end{theorem}

\begin{proof}
By considering data words as data trees
(e.g., by using only left-hand children starting from the root),
the lower bound follows from non-primitive recursiveness of
nonemptiness of one-way co-nondeterministic
(i.e., with only conjunctive branching)
automata with $1$ register over finite data words
\cite[Theorem~5.2]{Demri&Lazic09}.

We shall establish decidability by reducing to nonemptiness of ITCA,
which is decidable by Theorem~\ref{th:ITCA}.
More specifically, by extending to trees the translation in
the proof of \cite[Theorem~4.4]{Demri&Lazic09}, which is
from one-way alternating automata with $1$ register on finite data words
to incrementing counter automata on finite words,
we shall show that, for each ATRA$_1$ $\aregaut$,
an ITCA $\acaut_\aregaut$ with the same alphabet and such that
$\mathrm{L}(\acaut_\aregaut) =
 \{\mathrm{tree}(\adatatree) \,:\,
   \adatatree \in \mathrm{L}^{\mathit fin}(\aregaut)\}$,
is computable (in polynomial space).

Let $\aregaut = \tuple{\aalphabet, \locs, \aloc_I, F, \delta}$.
For a configuration $\aconf$ of $\aregaut$ and a datum $\adatum$,
let the \emph{bundle} of $\adatum$ in $\aconf$ be
the set of all states that are paired with $\adatum$,
i.e.\ $\{\aloc : \tuple{\aloc, \adatum} \in \aconf\}$.
The computation of $\acaut_\aregaut$ with the properties above
is based on the following abstraction of configurations of $\aregaut$
by mappings from $\mathcal{P}(\locs) \setminus \{\emptyset\}$ to $\mathbb{N}$.
The abstract configuration $\overline{\aconf}$ counts,
for each nonempty $\locster \subseteq \locs$,
the number of data whose bundles equal $\locster$:
\[\overline{\aconf}(\locster) =
  \length{\{\adatum \,:\,
            \{\aloc : \tuple{\aloc, \adatum} \in \aconf\} = \locster\}}\]
Thus, two configurations have the same abstraction iff
they are equal up to a bijective renaming of data.
For $1 \leq i \leq \overline{\aconf}(\locster)$ and $\aloc \in \locster$,
we shall call pairs $\tuple{\locster, i}$ \emph{abstract data}
and triples $\tuple{\aloc, \locster, i}$ \emph{abstract threads}.

For abstract configurations $\acval$, $\acvalbis_0$ and $\acvalbis_1$,
letters $\aletter$, and sets of states $\locs_=$ with
either $\acval(\locs_=) > 0$ or $\locs_= = \emptyset$,
we shall define transitions
$\acval
 \rightarrow_\aletter^{\locs_=}
 \acvalbis_0, \acvalbis_1$,
and show that they are bisimilar to transitions
$\aconf
 \rightarrow_\aletter^\adatumbis
 \aconfbis_0, \aconfbis_1$
such that $\acval = \overline{\aconf}$,
$\acvalbis_0 = \overline{\aconfbis_0}$,
$\acvalbis_1 = \overline{\aconfbis_1}$ and
$\locs_= = \{\aloc : \tuple{\aloc, \adatumbis} \in \aconf\}$.
The sets $\locs_=$ can hence be thought of as
abstractions of the data $\adatumbis$.
The abstract transitions will then give us a notion of
abstract run of $\aregaut$ on a finite tree (without data),
where the sets $\locs_=$ are guessed at every step.
By the bisimilarity, we shall have that:
\begin{describe}{(I)}
\item[(I)]
$\aregaut$ has an accepting abstract run on
a finite tree $\atree$ with alphabet $\aalphabet$ iff
it has an accepting run on some data tree $\adatatree$
such that $\atree = \mathrm{tree}(\adatatree)$.
\end{describe}
In other words, we shall have reduced the question of
whether $\mathrm{L}^{\mathit fin}(\aregaut)$ is nonempty,
i.e.\ whether there exists a finite tree with alphabet $\aalphabet$,
a data labelling of its nonleaf nodes,
and an accepting run of $\aregaut$ on the resulting data tree,
to whether there exists a finite tree and
an accepting abstract run of $\aregaut$ on it.
It will then remain to show how to compute (in polynomial space)
an ITCA $\acaut_\aregaut$ which guesses and checks
accepting abstract runs of $\aregaut$, so that:
\begin{describe}{(II)}
\item[(II)]
$\acaut_\aregaut$ has an accepting run on
a finite tree $\atree$ with alphabet $\aalphabet$ iff
$\aregaut$ has an accepting abstract run on $\atree$.
\end{describe}

To begin delivering our promises,
we now define transitions from abstract configurations $\acval$
for letters $\aletter$ and sets of states $\locs_=$ with
either $\acval(\locs_=) > 0$ or $\locs_= = \emptyset$
to abstract configurations $\acvalbis_0$ and $\acvalbis_1$,
essentially by reformulating the definition of concrete transitions
(cf.\ Section~\ref{ss:ATRA}) in terms of abstract threads.
For each abstract datum $\tuple{\locster, i}$ of $\acval$
and both $d \in \{0, 1\}$,
the abstract threads whose abstract datum is $\tuple{\locster, i}$
will contribute two sets of states to such a transition:
$\locsbis'(\locster, i)_d^\downarrow$,
for which the automaton's register is updated,
and $\locsbis'(\locster, i)_d^{\ndownarrow}$,
for which the automaton's register is not updated.
If $\locs_=$ is nonempty, we take $\tuple{\locs_=, 1}$
to represent the datum abstracted by $\locs^=$,
i.e.\ with which the register is updated,
so states in the union of the set $\locsbis'(\locs_=, 1)_d^{\ndownarrow}$
and all the sets $\locsbis'(\locster, i)_d^\downarrow$
will be associated to the same abstract datum of $\acvalbis_d$.
Formally, let 
$\acval
 \rightarrow_\aletter^{\locs_=}
 \acvalbis_0, \acvalbis_1$
mean that, for each abstract datum $\tuple{\locster, i}$ of $\acval$,
there exist sets of states
$\locsbis'(\locster, i)_0^\downarrow,
 \locsbis'(\locster, i)_0^{\ndownarrow},
 \locsbis'(\locster, i)_1^\downarrow,
 \locsbis'(\locster, i)_1^{\ndownarrow}$
such that:
\begin{itemize}
\item[(i)]
for each abstract thread $\tuple{\aloc, \locster, i}$ of $\acval$,
there exist
\[\locsbis_0^\downarrow,
  \locsbis_0^{\ndownarrow},
  \locsbis_1^\downarrow,
  \locsbis_1^{\ndownarrow}
  \models
  \delta(\aloc, \aletter, \tuple{\locster, i} = \tuple{\locs_=, 1})\]
which satisfy
$\locsbis_d^? \subseteq
 \locsbis'(\locster, i)_d^?$
for both $d \in \{0, 1\}$ and $? \in \{\downarrow, \ndownarrow\}$;
\item[(ii)]
for both $d \in \{0, 1\}$ and each nonempty $\locster' \subseteq \locs$,
we have
\[\length{\{\tuple{\locster, i} \,:\,
            \tuple{\locster, i} \neq \tuple{\locs_=, 1} \,\wedge\,
            \locsbis'(\locster, i)_d^{\ndownarrow} = \locster'\}} +
  \left\{\begin{array}{ll}
  1, & \mathrm{if}\ \locsbis^=_d = \locster' \\
  0, & \mathrm{otherwise}
  \end{array}\right\}
  \leq \acvalbis_d(\locster')\]
for some
$\locsbis^=_d \,\supseteq\,
 \locsbis'(\locs_=, 1)_d^{\ndownarrow} \,\cup\,
 \bigcup_{1 \leq i \leq \acval(\locster)} \locsbis'(\locster, i)_d^\downarrow$.
\end{itemize}

It is straightforward to check the following two-part correspondence between
the abstract transitions just defined and concrete transitions:
\begin{describe}{(IIIb)}
\item[(IIIa)]
Whenever
$\aconf
 \rightarrow_\aletter^\adatumbis
 \aconfbis_0, \aconfbis_1$,
we have
$\acval
 \rightarrow_\aletter^{\locs_=}
 \acvalbis_0, \acvalbis_1$,
where $\acval = \overline{\aconf}$,
$\acvalbis_0 = \overline{\aconfbis_0}$,
$\acvalbis_1 = \overline{\aconfbis_1}$ and
$\locs_= = \{\aloc : \tuple{\aloc, \adatumbis} \in \aconf\}$.
\item[(IIIb)]
Whenever $\overline{\aconf} = \acval$ and
$\acval
 \rightarrow_\aletter^{\locs_=}
 \acvalbis_0, \acvalbis_1$,
there exist $\adatumbis$, $\aconfbis_0$ and $\aconfbis_1$ such that
$\aconf
 \rightarrow_\aletter^\adatumbis
 \aconfbis_0, \aconfbis_1$,
$\acvalbis_0 = \overline{\aconfbis_0}$,
$\acvalbis_1 = \overline{\aconfbis_1}$ and
$\locs_= = \{\aloc : \tuple{\aloc, \adatumbis} \in \aconf\}$.
\end{describe}
Let $\alpha$ be a bijection between the abstract data of $\acval$
and the data that occur in $\aconf$, which is bundle preserving
(i.e., whenever $\alpha \tuple{\locster, i} = \adatum$,
we have that $\locster$ is the bundle of $\adatum$ in $\aconf$),
and if $\locs_=$ is nonempty then $\alpha \tuple{\locs_=, 1} = \adatumbis$.
\begin{itemize}
\item
To show (IIIa),
for each abstract datum $\tuple{\locster, i}$ of $\acval$
and both $d \in \{0, 1\}$,
take $\locsbis'(\locster, i)_d^\downarrow$ and $\locsbis^=_d$ to be
the bundle of $\adatumbis$ in $\aconfbis_d$,
and take $\locsbis'(\locster, i)_d^{\ndownarrow}$ to be
the bundle of $\alpha \tuple{\locster, i}$ in $\aconfbis_d$.
\item
For (IIIb), if $\locs_=$ is empty then take $\adatumbis$ to be
an arbitrary datum which does not occur in $\aconf$,
pick the same quadruples for the threads in $\aconf$
as for the corresponding (via $\alpha$) abstract threads of $\acval$,
and for both $d \in \{0, 1\}$, obtain $\aconfbis_d$ from $\acvalbis_d$
by replacing each set of abstract data
$\tuple{\locster', 1}$, \ldots,
$\tuple{\locster', \acvalbis_d(\locster')}$ with:
the data $\alpha \tuple{\locster, i}$ such that
$\tuple{\locster, i} \neq \tuple{\locs_=, 1}$ and
$\locsbis'(\locster, i)_d^{\ndownarrow} = \locster'$,
the datum $\adatumbis$ if $\locsbis^=_d = \locster'$,
and fresh further data if the inequality in (ii) is strict.
\end{itemize}

Composing abstract transitions gives us abstract runs of $\aregaut$.
Such a run on a finite tree $\tuple{\nodes, \aalphabet, \Lambda}$
is a mapping $\anode \mapsto \acval_\anode$
from the nodes to abstract configurations such that,
for each nonleaf $\anode$, there exists $\locs_=$ with
$\acval_\anode
 \rightarrow_{\Lambda(\anode)}^{\locs_=}
 \acval_{\anode \cdot 0}, \acval_{\anode \cdot 1}$,
and if $\anode$ is the root then $\aloc_I \in \locs_=$.
Defining the run to be accepting iff $\acval_\anode(\locster) = 0$
for all leaves $\anode$ and all $\locster \not\subseteq F$,
we have (I) above by (IIIa) and (IIIb).

We are now ready to define $\acaut_\aregaut$,
as an ITCA that performs the steps (1)--(9) below.
States of $\acaut_\aregaut$ are used for control and for storing
$\aletter$, $\locs_=$, $\mathit{root}$ (initially $\mathit{tt}$), $\locster$,
${\locsbis'}_0^\downarrow$,
${\locsbis'}_0^{\ndownarrow}$,
${\locsbis'}_1^\downarrow$,
${\locsbis'}_1^{\ndownarrow}$,
$\aloc$,
$\locsbis_0^\downarrow$,
$\locsbis_0^{\ndownarrow}$,
$\locsbis_1^\downarrow$,
$\locsbis_1^{\ndownarrow}$,
$d$, $?$ and
$\locsbis^=_d$.
There are $2^{\length{\locs}} - 1$ counters in the array $c$,
and $2^{\length{\locs}^4}$ counters in the array $c'$.
The steps are implemented by $\emptyword$-transitions,
except for the $\aletter$-transition in (4).
The choices are nondeterministic.
If a choice in (3.2) is impossible,
or a check in (2), (3.2) or (5) fails,
then $\acaut_\aregaut$ blocks.

The steps (1)--(9) guess and check
an accepting abstract run of $\aregaut$ on a finite tree.
The counter array $c$ is used to store abstract configurations,
and the counter array $c'$ is auxiliary.
The initial condition in the definition of abstract runs is checked in (2),
the final condition in (8), and steps (3)--(7) are essentially
a reformulation of the definition of abstract transitions.
This particular reformulation is tailored for
a development in the proof of Theorem~\ref{th:safety},
and is based on observing that the quadruples of sets
$\locsbis'(\locster, i)_0^\downarrow,
 \locsbis'(\locster, i)_0^{\ndownarrow},
 \locsbis'(\locster, i)_1^\downarrow,
 \locsbis'(\locster, i)_1^{\ndownarrow}$
for abstract data $\tuple{\locster, i} \neq \tuple{\locs_=, 1}$
do not need to be stored simultaneously,
i.e.\ that it suffices to store numbers of such identical quadruples,
which is done using the counter array $c'$.
\begin{describe}{(1)}
\item[(1)]
Choose $\aletter \in \aalphabet$, and $\locs_=$ with
either $c[\locs_=] > 0$ or $\locs_= = \emptyset$.
\item[(2)]
If $\mathit{root} = \mathit{tt}$,
then check that $\aloc_I \in \locs_=$ and set $\mathit{root} := \mathit{ff}$.
\item[(3)]
For each nonempty $\locster \subseteq \locs$,
while $c[\locster] > 0$ do:
\begin{describe}{(3.1)}
\item[(3.1)]
choose
${\locsbis'}_0^\downarrow,
 {\locsbis'}_0^{\ndownarrow},
 {\locsbis'}_1^\downarrow,
 {\locsbis'}_1^{\ndownarrow}
 \subseteq \locs$;
\item[(3.2)]
for each $\aloc \in \locster$,
choose
$\locsbis_0^\downarrow,
 \locsbis_0^{\ndownarrow},
 \locsbis_1^\downarrow,
 \locsbis_1^{\ndownarrow}
 \models
 \delta(\aloc, \aletter, \tuple{\locster, c[\locster]} = \tuple{\locs_=, 1})$,
and check that
$\locsbis_d^? \subseteq
 {\locsbis'}_d^?$
for both $d \in \{0, 1\}$ and $? \in \{\downarrow, \ndownarrow\}$;
\item[(3.3)]
decrement $c[\locster]$, and
if $\tuple{\locster, c[\locster]} = \tuple{\locs_=, 0}$,
then choose $\locsbis^=_d \supseteq {\locsbis'}_d^\downarrow \cup
                                    {\locsbis'}_d^{\ndownarrow}$
for both $d \in \{0, 1\}$,
else increment $c'[{\locsbis'}_0^\downarrow,
                   {\locsbis'}_0^{\ndownarrow},
                   {\locsbis'}_1^\downarrow,
                   {\locsbis'}_1^{\ndownarrow}]$.
\end{describe}
\item[(4)]
Perform an $\aletter$-transition,
forking with $d := 0$ and $d := 1$.
\item[(5)]
Check that $\locsbis^=_d \supseteq
            \bigcup \{{\locsbis'}_d^\downarrow :
                      c'[{\locsbis'}_0^\downarrow,
                         {\locsbis'}_0^{\ndownarrow},
                         {\locsbis'}_1^\downarrow,
                         {\locsbis'}_1^{\ndownarrow}] > 0\}$,
and increment $c[\locsbis^=_d]$.
\item[(6)]
Transfer each $c'[{\locsbis'}_0^\downarrow,
                  {\locsbis'}_0^{\ndownarrow},
                  {\locsbis'}_1^\downarrow,
                  {\locsbis'}_1^{\ndownarrow}]$
with nonempty ${\locsbis'}_d^{\ndownarrow}$
to $c[{\locsbis'}_d^{\ndownarrow}]$.
\item[(7)]
Reset (i.e.\ decrement until $0$) each $c'[{\locsbis'}_0^\downarrow,
                                           {\locsbis'}_0^{\ndownarrow},
                                           {\locsbis'}_1^\downarrow,
                                           {\locsbis'}_1^{\ndownarrow}]$
with empty ${\locsbis'}_d^{\ndownarrow}$.
\item[(8)]
If $c[\locster] = 0$ whenever $\locster \not\subseteq F$,
then pass through a final state.
\item[(9)]
Repeat from (1).
\end{describe}

Since $\acaut_\aregaut$ is an ITCA,
its runs may contain arbitrary errors that increase one or more counters.
Nevertheless, between executions of steps (3)--(7) by $\acaut_\aregaut$
and abstract transitions of $\aregaut$,
we have the following two-part correspondence.
It shows that the possibly erroneous executions of (3)--(7)
match the abstract transitions with the slack allowed by their definition,
which in turn match the concrete transitions
with their possible introductions of junk threads
(cf.\ (IIIa), (IIIb) and Remark~\ref{r:junk}).
\begin{describe}{(IVb)}
\item[(IVa)]
Whenever
$\acval
 \rightarrow_\aletter^{\locs_=}
 \acvalbis_0, \acvalbis_1$,
we have that $\acaut_\aregaut$ can perform steps (3)--(7)
beginning with any configuration such that
each $c[\locster]$ has value $\acval(\locster)$ and
each $c'[{\locsbis'}_0^\downarrow,
         {\locsbis'}_0^{\ndownarrow},
         {\locsbis'}_1^\downarrow,
         {\locsbis'}_1^{\ndownarrow}]$
has value $0$,
so that for both forks $d \in \{0, 1\}$ in (4),
the ending configuration is such that
each $c[\locster]$ has value $\acvalbis_d(\locster)$ and
each $c'[{\locsbis'}_0^\downarrow,
         {\locsbis'}_0^{\ndownarrow},
         {\locsbis'}_1^\downarrow,
         {\locsbis'}_1^{\ndownarrow}]$
has value $0$.
\item[(IVb)]
Whenever $\acaut_\aregaut$ can perform steps (3)--(7)
beginning with a configuration such that
$\aletter$ and $\locs_=$ are as in (1) and
each $c[\locster]$ has value $\acval(\locster)$,
so that for both forks $d \in \{0, 1\}$ in (4),
the ending configuration is such that
each $c[\locster]$ has value $\acvalbis_d(\locster)$,
we have
$\acval
 \rightarrow_\aletter^{\locs_=}
 \acvalbis_0, \acvalbis_1$.
\end{describe}
\begin{itemize}
\item
In proving (IVa), we can choose where incrementing errors occur.
For each iteration of (3.1)--(3.3),
let the quadruple chosen in (3.1) be
\[\locsbis'(\locster, c[\locster])_0^\downarrow,
  \locsbis'(\locster, c[\locster])_0^{\ndownarrow},
  \locsbis'(\locster, c[\locster])_1^\downarrow,
  \locsbis'(\locster, c[\locster])_1^{\ndownarrow}\]
so that (3.2) can succeed by (i) in the definition of abstract transitions.
It remains to match by incrementing errors, say at the end of (7),
any differences between the two sides of the inequalities in (ii).
\item
To obtain (IVb), let
$\locsbis'(\locster, i)_0^\downarrow,
 \locsbis'(\locster, i)_0^{\ndownarrow},
 \locsbis'(\locster, i)_1^\downarrow,
 \locsbis'(\locster, i)_1^{\ndownarrow}$
for each abstract datum $\tuple{\locster, i}$ of $\acval$
be the quadruple chosen in the last performance of (3.1) with $i = c[\locster]$
(due to incrementing errors, there may be more than one).
Step~(3.2) ensures that (i) is satisfied.
Since at the end of (3),
each $c'[{\locsbis'}_0^\downarrow,
         {\locsbis'}_0^{\ndownarrow},
         {\locsbis'}_1^\downarrow,
         {\locsbis'}_1^{\ndownarrow}]$
has value at least
\[\length{\{\tuple{\locster, i} \,:\,
            \tuple{\locster, i} \neq \tuple{\locs_=, 1} \,\wedge\,
            \forall d, ? (\locsbis'(\locster, i)_d^? = {\locsbis'}_d^?)\}}\]
steps (5) and (6) ensure that (ii) is satisfied.
\end{itemize}

Now, we have (II) above.
The `if' direction follows by (IVa),
and the `only if' direction by (IVb)
once we observe that, without loss of generality,
we can consider only runs of $\acaut_\aregaut$ that do not
contain incrementing errors on the array $c$ outside of steps (3)--(7)
except before the first performance of (1).

To conclude that polynomial space suffices for computing $\acaut_\aregaut$,
we observe that each of its state variables is
either from a fixed finite set, or an element of $\aalphabet$,
or an element or subset of $\locs$,
and that deciding satisfaction of transition formulae
$\delta(\aloc, \aletter, \tuple{\locster, c[\locster]} = \tuple{\locs_=, 1})$
in step~(3.2) amounts to evaluating Boolean formulae.
\end{proof}

We remark that, in the opposite direction to
the translation in the proof of Theorem~\ref{th:ATRA1},
by extending the proof of \cite[Theorem~5.2]{Demri&Lazic09} to trees,
for each ITCA $\acaut$, an ATRA$_1$ $\aregaut_\acaut$
is computable in logarithmic space such that
$\mathrm{L}^{\mathit fin}(\aregaut_\acaut)$ consists of
encodings of accepting runs of $\acaut$.
Moreover, similarly as on words,
the two translations can be extended to infinite trees,
where ATRA$_1$ are equipped with weak acceptance
and ITCA with B\"uchi acceptance.
Instead of decidable and not primitive recursive as on finite trees,
nonemptiness for those two classes of automata can then be shown
co-r.e.-complete.

\section{Safety Automata}
\label{s:safety}

We now show decidability of nonemptiness of
forward alternating tree $1$-register automata
with safety acceptance over finite or infinite data trees.
More precisely, since the class of safety ATRA$_1$
is not closed under complement,
but is closed under intersection and union
(cf.\ Proposition~\ref{pr:closure.RA}(b)),
we show decidability of the inclusion problem,
which implies decidability of nonemptiness of
Boolean combinations of safety ATRA$_1$.
However, already for the subproblems of nonemptiness and nonuniversality,
we obtain non-elementary and non-primitive recursive lower bounds
(respectively).

%As straightforward consequences of the proofs of
%$\Sigma^0_1$-hardness of finitary nonemptiness for
%alternating word register automata
%which either are two-way or have two registers
%\cite[Theorem~5.4]{Demri&Lazic09},
%we have that nonemptiness and nonuniversality
%are undecidable for the safety subsets of
%ATRA$^1(\odown, \oup)$, ATRA$^1(\oright, \oleft)$,
%ATRA$^2(\odown)$ and ATRA$^2(\oright)$.

\begin{theorem}
\label{th:safety}
For safety ATRA$_1$,
inclusion is decidable,
nonemptiness is not elementary, and
nonuniversality is not primitive recursive.
\end{theorem}

\begin{proof}
Showing that the inclusion problem is decidable will involve extending:
\begin{itemize}
\item
the proof of Proposition~\ref{pr:closure.RA} to obtain
an intersection of a safety and a co-safety ATRA$_1$,
which can be seen as a weak parity ATRA$_1$ with $2$ priorities;
\item
the proof of Theorem~\ref{th:ATRA1} to obtain
an ITCA with a more powerful set of instructions
and no cycles of $\emptyword$-transitions,
which can also be seen as having weak parity acceptance with $2$ priorities;
\item
the proof of Theorem~\ref{th:ITCA} to obtain
decidability of nonemptiness of such ITCA.
\end{itemize}
To maintain focus,
we shall avoid introducing the extended notions in general,
but concentrate on what is necessary for this part of the proof.

Suppose
$\aregaut_1 = \tuple{\aalphabet, \locs_1, \aloc_I^1, F_1, \delta_1}$ and
$\aregaut_2 = \tuple{\aalphabet, \locs_2, \aloc_I^2, F_2, \delta_2}$
are ATRA$_1$, where we need to determine whether
$\mathrm{L}^{\mathit saf}(\aregaut_1)$ is a subset of
$\mathrm{L}^{\mathit saf}(\aregaut_2)$.
By the proof of Proposition~\ref{pr:closure.RA}(b),
that amounts to emptiness of the intersection of
$\mathrm{L}^{\mathit saf}(\aregaut_1)$ and
$\mathrm{L}^{\mathit cos}(\overline{\aregaut_2})$, where
$\overline{\aregaut_2} =
 \tuple{\aalphabet, \locs_2, \aloc_I^2, \overline{F_2}, \overline{\delta_2}}$
is the dual automaton to $\aregaut_2$.
Assuming that $\locs_1$ and $\locs_2$ are disjoint,
and do not contain $\aloc_I^\cap$, let
\[\aregaut_\cap =
  \tuple{\aalphabet, \{\aloc_I^\cap\} \cup \locs_1 \cup \locs_2,
         \aloc_I^\cap, F_1 \cup F_2, \delta_\cap}\]
be the automaton for the intersection of
$\aregaut_1$ and $\overline{\aregaut_2}$:
\[\delta_\cap(\aloc, \aletter, p) =
\left\{\begin{array}{ll}
\delta(\aloc_I^1, \aletter, p) \wedge
\delta(\aloc_I^2, \aletter, p), &
\mathrm{if}\ \aloc = \aloc_I^\cap
\\
\delta_1(\aloc, \aletter, p), &
\mathrm{if}\ \aloc \in \locs_1
\\
\overline{\delta_2(\aloc, \aletter, p)}, &
\mathrm{if}\ \aloc \in \locs_2
\end{array}\right.\]
We then have:
\begin{describe}{(*)}
\item[(*)]
A data tree $\adatatree$ with alphabet $\aalphabet$ is
safety-accepted by $\aregaut_1$ and
co-safety-accepted by $\overline{\aregaut_2}$
iff $\aregaut_\cap$ has a run on $\adatatree$
which is final and $\locs_2$-finite,
i.e.\ there exists $l$ such that
the configuration at each node of length at least $l$
contains no threads with states from $\locs_2$.
\end{describe}

Before proceeding, let
\emph{incrementing tree counter automata with nondeterministic transfers}
(shortly, \emph{ITCANT}) be defined as ITCA (cf.\ Section~\ref{ss:ITCA}),
except that $\tuple{\mathtt{ifz}, c}$ instructions are replaced by
$\tuple{\mathtt{transf}, c, C}$ for counters $c$ and sets of counters $C$.
Such an instruction is equivalent to a loop which executes while $c$ is nonzero,
and in each iteration, decrements $c$ and increments some counter in $C$.
However, in presence of incrementing errors, the loop may not terminate,
whereas $\tuple{\mathtt{transf}, c, C}$ instructions are considered atomic.
The effect of $\tuple{\mathtt{transf}, c, C}$ is therefore
to transfer the value of $c$ to the counters in $C$,
among which it is split nondeterministically.
In particular, $\tuple{\mathtt{ifz}, c}$ instructions
can be reintroduced as $\tuple{\mathtt{transf}, c, \emptyset}$.

Now, steps (1)--(9) in the proof of Theorem~\ref{th:ATRA1}
can be implemented by an ITCANT which uses nondeterministic transfers
instead of the loops in (3), (6) and (7), and whose transition relation
therefore contains no cycles of $\emptyword$-transitions.
More specifically, each reset in (7) can be implemented as
a transfer to a new auxiliary counter $c''$,
(6) already consists of transfers to single counters,
and (3) can be replaced by the following two steps:
\begin{describe}{(3b)}
\item[(3a)]
If $\locs_= \neq \emptyset$,
then decrement $c[\locs_=]$
and choose $\locsbis^=_0,
            \locsbis^=_1
            \subseteq \locs$
such that, for each $\aloc \in \locs_=$,
there exist
$\locsbis_0^\downarrow,
 \locsbis_0^{\ndownarrow},
 \locsbis_1^\downarrow,
 \locsbis_1^{\ndownarrow}
 \models
 \delta(\aloc, \aletter, \mathit{tt})$
with
$\locsbis^=_d \supseteq \locsbis_d^\downarrow \cup
                        \locsbis_d^{\ndownarrow}$
for both $d \in \{0, 1\}$.
\item[(3b)]
Transfer each $c[\locster]$ nondeterministically
to the set of all $c'[{\locsbis'}_0^\downarrow,
                      {\locsbis'}_0^{\ndownarrow},
                      {\locsbis'}_1^\downarrow,
                      {\locsbis'}_1^{\ndownarrow}]$
such that, for each $\aloc \in \locster$,
there exist
$\locsbis_0^\downarrow,
 \locsbis_0^{\ndownarrow},
 \locsbis_1^\downarrow,
 \locsbis_1^{\ndownarrow}
 \models
 \delta(\aloc, \aletter, \mathit{ff})$
with
$\locsbis_d^? \subseteq
 {\locsbis'}_d^?$
for both $d \in \{0, 1\}$ and $? \in \{\downarrow, \ndownarrow\}$.
\end{describe}
Let $\acaut_\cap$ be such an ITCANT for $\aregaut_\cap$,
which in addition performs the following step between (7) and (8),
where $\mathit{prop}$ is a state variable, initially $\mathit{ff}$:
\begin{describe}{(7$\frac{1}{2}$)}
\item[(7$\frac{1}{2}$)]
If $c[\locster] = 0$ whenever $\locster \cap \locs_2 \neq \emptyset$,
then set $\mathit{prop} := \mathit{tt}$.
\end{describe}
As in the proof of Theorem~\ref{th:ATRA1},
we have that $\acaut_\cap$ is computable from $\aregaut_\cap$,
and therefore from $\aregaut_1$ and $\aregaut_2$, in polynomial space.
Also, $\aregaut_\cap$ satisfies (IIIa) and (IIIb),
and $\aregaut_\cap$ and $\acaut_\cap$ satisfy (IVa) and (IVb).
Recalling that $\acaut_\cap$ contains no cycles of $\emptyword$-transitions,
we infer the following from (*) above,
where the notion of transitions between levels of $\acaut_\cap$
is as in the proof of Theorem~\ref{th:ITCA},
and $P$ denotes the set of all states of $\acaut_\cap$
in which $\mathit{prop}$ has value $\mathit{tt}$:
\begin{describe}{(**)}
\item[(**)]
$\mathrm{L}^{\mathit saf}(\aregaut_1)$ is a subset of
$\mathrm{L}^{\mathit saf}(\aregaut_2)$
iff there does not exist an infinite sequence of transitions
$\aconfs_0 \rightarrow \aconfs_1 \rightarrow \cdots$
which is from the initial level of $\acaut_\cap$ and such that
some $\aconfs_i$ contains only states from $P$.
\end{describe}

To conclude decidability of inclusion, we show that,
given an ITCANT $\acaut_\cap$ and a set $P$ of its states,
existence of an infinite sequence of transitions as in (**) is decidable.
For a set $\aconfss$ of levels of $\acaut_\cap$,
we write ${\uparrow} \aconfss$ to denote
its upward closure with respect to $\preceq$:
the set of all $\aconfsbis$ for which there exists
$\aconfs \in \aconfss$ with $\aconfs \preceq \aconfsbis$.
We say that $\aconfss$ is upwards closed
iff $\aconfss = {\uparrow} \aconfss$, and
we say that $\aconfssbis$ is a basis for $\aconfss$
iff $\aconfss = {\uparrow} \aconfssbis$.
As in the proof of Theorem~\ref{th:ITCA}, we have that
successor sets with respect to $\rightarrow$ are computable,
$\preceq$ is a well-quasi-ordering,
$\preceq$ is strongly (in particular, reflexively)
downward-compatible with $\rightarrow$, and
$\preceq$ is decidable.
Hence, by \cite[Proposition~5.4]{Finkel&Schnoebelen01},
a finite basis $\aconfss_R$ of the upward closure of
the set of all levels reachable from the initial level is computable.
By the strong downward compatibility,
the set of all levels from which
there exists an infinite sequence of transitions
is downwards closed, so its complement is upwards closed.
We claim that a finite basis $\aconfss_T$ of the latter set is computable.
With that assumption, since a finite basis $\aconfss_N$ of the set of
all levels that contain some state not from $P$ is certainly computable,
we are done because there does not exist
an infinite sequence of transitions as in (**)
iff ${\uparrow} \aconfss_R$ is a subset of the union of
${\uparrow} \aconfss_T$ and ${\uparrow} \aconfss_N$.

It remains to establish the claim.
For a finite set $\aconfss'$ of levels of $\acaut_\cap$, let
\[K(\aconfss') =
  1 +
  \max_{\aconfs' \in \aconfss'}
  \max_{\tuple{\aloc, \acval} \in \aconfs'}
  \sum_{c \in \{1, \ldots, k\}} \acval(c)\]
where $k$ is the number of counters of $\acaut_\cap$.
Let also $\mathrm{Pred}_\forall(\aconfss')$ be the upwards-closed set
consisting of all $\aconfs$ such that,
whenever $\aconfs \rightarrow \aconfs'$,
we have $\aconfs' \in {\uparrow} \aconfss'$.
Observe that, whenever $\aconfs \in \mathrm{Pred}_\forall(\aconfss')$,
there exists $\aconfs_\dag \in \mathrm{Pred}_\forall(\aconfss')$
such that $\aconfs_\dag \preceq \aconfs$ and, for each
$\tuple{\aloc, \acval} \in \aconfs_\dag$ and $c \in \{1, \ldots, k\}$,
$\acval(c) \leq K(\aconfss')$.
Hence, a finite basis of $\mathrm{Pred}_\forall(\aconfss')$ is computable,
so the following is an effective procedure:
\begin{describe}{(iii)}
\item[(i)]
Begin with $\aconfss_T := \emptyset$.
\item[(ii)]
Let $\aconfssbis$ be a finite basis of $\mathrm{Pred}_\forall(\aconfss_T)$.
\item[(iii)]
If $\aconfssbis \not\subseteq {\uparrow} \aconfss_T$,
then set $\aconfss_T := \aconfss_T \cup \aconfssbis$ and repeat from (ii),
else terminate.
\end{describe}
Since $\preceq$ is a well-quasi-ordering,
the procedure terminates and computes a basis of
the set of all levels from which
every sequence of transitions is finite,
as required.

We shall establish that nonemptiness of safety ATRA$_1$
is not elementary by a two-stage reduction, which separates
dealing with the inability of one-way alternating $1$-register automata
to detect incrementing errors in encodings of computations of counter machines,
from ensuring acceptance only of encodings of computations
in which counters are bounded by a tower of exponentiations.
More precisely, we shall use the following problem as intermediary.
The notation $2 \Uparrow m$ is as in Example~\ref{ex:ATRA}.
\begin{describe}{(***)}
\item[(***)]
Given a deterministic counter machine $\acaut$ and $m \geq 1$ in unary,
does $\acaut$ have a computation which possibly contains incrementing errors,
in which every counter value is at most $2 \Uparrow m$,
and which is either halting or infinite?
\end{describe}
Such a machine is a tuple $\tuple{\locs, \aloc_I, \aloc_H, k, \delta}$ where:
$\locs$ is a finite set of states,
$\aloc_I$ is the initial state,
$\aloc_H$ is the halting state,
$k \in \mathbb{N}$ is the number of counters, and
$\delta: \locs \setminus \{\aloc_H\} \,\rightarrow\,
         \{1, \ldots, k\} \times (\locs \cup \locs^2)$
is a transition function.
Thus, from a state $\aloc \neq \aloc_H$,
either $\delta(\aloc)$ is of the form $\tuple{c, \aloc'}$,
which means that the machine increments $c$ and goes to $\aloc'$,
or $\delta(\aloc)$ is of the form $\tuple{c, \aloc', \aloc''}$,
which means that, if $c$ is zero, then the machine goes to $\aloc'$,
else it decrements $c$ and goes to $\aloc''$.
More precisely, a configuration is a state together with a counter valuation,
and we write $\tuple{\aloc, \acval} \rightarrow \tuple{\aloc', \acval'}$ iff,
for some $\acval_\surd \geq \acval$ and $\acval'_\surd \leq \acval'$,
\begin{itemize}
\item
either $\delta(\aloc) = \tuple{c, \aloc'}$ and
$\acval'_\surd = \acval_\surd[c \mapsto \acval_\surd(c) + 1]$,
\item
or $\delta(\aloc) = \tuple{c, \aloc', \aloc''}$,
$\acval_\surd(c) = 0$ and $\acval'_\surd = \acval_\surd$,
\item
or $\delta(\aloc) = \tuple{c, \aloc'', \aloc'}$ and
$\acval'_\surd = \acval_\surd[c \mapsto \acval_\surd(c) - 1]$.
\end{itemize}
We say that the transition is error-free iff
the above holds with $\acval_\surd = \acval$ and $\acval'_\surd = \acval'$.
A computation is a sequence
$\tuple{\aloc_0, \acval_0} \rightarrow
 \tuple{\aloc_1, \acval_1} \rightarrow
 \cdots$
such that $\aloc_0 = \aloc_I$ and $\acval = \mathbf{0}$.

To show that (***) is not elementary, we reduce from
the problem of whether a deterministic $2$-counter machine of size $m$
has an error-free halting computation of length at most $2 \Uparrow m$.
Given such a machine $\acaut$ whose counters are $c_1$ and $c_2$,
let $\widehat{\acaut}$ be a deterministic machine with counters
$c_1$, $c_2$, $\overline{c_1}$, $\overline{c_2}$,
$c^\dag$, $c'$, $c''$ and $c'''$,
which performs the following and then halts:
\begin{describe}{(III)}

\item[(I)]
For both $i \in \{1, 2\}$, set $\overline{c_i}$ to $2 \Uparrow m$
by executing the pseudo-code in Figure~\ref{f:tetration}.
The loops over $c'$, $c''$ and $c'''$ implement
$\overline{c_i} := 2 \Uparrow c'$,
$\overline{c_i} := 2^{c''}$ and
$\overline{c_i} := 2 \times c'''$
(respectively).

\item[(II)]
Simulate $\acaut$ using $c_1$ and $c_2$, and after each step:
\begin{itemize}
\item
increment $c^\dag$;
\item
if $c_i$ has been incremented, then decrement $\overline{c_i}$;
\item
if $c_i$ has been decremented, then increment $\overline{c_i}$;
\item
if $\acaut$ has halted, then go to (III).
\end{itemize}

\item[(III)]
For both $i \in \{1, 2\}$, transfer $\overline{c_i}$ to $c_i$.
\end{describe}

\begin{narrowfig}{.67\textwidth}
\begin{tabbing}
$c' := m$;
$\mathtt{inc}(\overline{c_i})$; \\
\textbf{while} $c' > 0$ \\
\{ \= $\mathtt{dec}(c')$;
      \textbf{while} $\overline{c_i} > 0$
      \{ $\mathtt{dec}(\overline{c_i})$;
         $\mathtt{inc}(c'')$ \};
      $\mathtt{inc}(\overline{c_i})$; \\
   \> \textbf{while} $c'' > 0$ \\
   \> \{ \= $\mathtt{dec}(c'')$;
            \textbf{while} $\overline{c_i} > 0$
            \{ $\mathtt{dec}(\overline{c_i})$;
               $\mathtt{inc}(c''')$ \}; \\
   \>    \> \textbf{while} $c''' > 0$
            \{ $\mathtt{dec}(c''')$;
               $\mathtt{inc}(\overline{c_i})$;
               $\mathtt{inc}(\overline{c_i})$ \} \} \}
\end{tabbing}
\caption{Computing $2 \Uparrow m$}
\label{f:tetration}
\end{narrowfig}

Observe that $\widehat{\acaut}$ is computable in space logarithmic in $m$.
If $\acaut$ has an error-free halting computation
of length at most $2 \Uparrow m$,
running $\widehat{\acaut}$ without errors indeed halts and
does not involve counter values greater than $2 \Uparrow m$.
For the converse, suppose $\widehat{\acaut}$ has a computation
which possibly contains incrementing errors,
in which every counter value is at most $2 \Uparrow m$,
and which is either halting or infinite.
By the construction of $\widehat{\acaut}$
and the boundedness of counter values,
the computation cannot be infinite, so it is halting.
Since $\overline{c_1}$ and $\overline{c_2}$
were set to $2 \Uparrow m$ by stage (I),
and since stage (III) terminated,
the halting computation of $\acaut$ in stage (II) must have been error-free
and it is certainly of length at most $2 \Uparrow m$.

To reduce from (***) to nonemptiness of safety ATRA$_1$,
consider a deterministic counter machine
$\acaut = \tuple{\locs, \aloc_I, \aloc_H, k, \delta}$
and $m \geq 1$.
We can assume that $\aloc' \neq \aloc''$
whenever $\delta(\aloc) = \tuple{c, \aloc', \aloc''}$.
By the proof of \cite[Theorem~5.2]{Demri&Lazic09},
which uses essentially the same encoding of
computations of counter machines into data words
as in the proof of \cite[Theorem~14]{Bojanczyketal06a},
we have that an ATRA$_1$ $\aregaut_\acaut$ with alphabet $\locs$
is computable in space logarithmic in $\length{\acaut}$,
such that it safety-accepts a data tree $\adatatree$
iff the left-most path in $\adatatree$
(i.e., the sequence of nodes obtained by
starting from the root and repeatedly taking the left-hand child)
satisfies the following:
\begin{itemize}

\item
the letter of the first node is $\aloc_I$, and
either the letter of the last nonleaf node is $\aloc_H$
or the sequence is infinite;

\item
for all letters $\aloc$ and $\aloc'$ of
two consecutive nodes $\anode$ and $\anode'$ (respectively),
\begin{itemize}
\item
either $\delta(\aloc)$ is of the form $\tuple{c, \aloc'}$
and we say that $\anode$ is $c$-decrementing,
\item
or $\delta(\aloc)$ is of the form $\tuple{c, \aloc', \aloc''}$
and we say that $\anode$ is $c$-zero-testing,
\item
or $\delta(\aloc)$ is of the form $\tuple{c, \aloc'', \aloc'}$
and we say that $\anode$ is $c$-decrementing;
\end{itemize}

\item
for each counter $c$,
no two $c$-incrementing nodes are labelled by the same datum,
no two $c$-decrementing nodes are labelled by the same datum,
and whenever a $c$-incrementing node $\anode$
is followed by a $c$-zero-testing node $\anode'$,
then a $c$-decrementing node with the same datum as $\anode$
must occur between $\anode$ and $\anode'$.
\end{itemize}
Hence, by taking the left-most paths in data trees
that are safety-accepted by $\aregaut_\acaut$ and erasing data,
we obtain exactly the sequences of states of halting or infinite
computations of $\acaut$ which possibly contain incrementing errors.
Assuming that $\aletterbis_1$, \ldots, $\aletterbis_m$, $*$ are not in $\locs$,
to restrict further to computations of $\acaut$
in which every counter value is at most $2 \Uparrow m$,
it suffices to strengthen $\aregaut_\acaut$ to obtain
a safety ATRA$_1$ $\aregaut_\acaut^{2 \Uparrow m}$
with alphabet $\locs \cup \{\aletterbis_1, \ldots, \aletterbis_m, *\}$
which requires that:
\begin{itemize}
\item
whenever a node $\anode$ in the left-most path is $c$-incrementing,
then the automaton $\aregautbis_m$ from Example~\ref{ex:ATRA}
safety-accepts the right-hand subtree at $\anode$;
\item
whenever a node $\anode$ in the left-most path is $c$-incrementing,
$\anode'$ is either $\anode$ or a subsequent $c$-incrementing node,
and no $c$-decrementing node with the same datum as $\anode$
occurs between $\anode$ and $\anode'$,
then the right-hand subtree at $\anode'$ contains
a node with letter $\aletterbis_m$ and the same datum as $\anode$.
\end{itemize}

Finally, that nonuniversality of safety ATRA$_1$ is not primitive recursive
follows from the same lower bound for nonuniversality of
safety one-way alternating automata with $1$ register over data words
\cite{Lazic06}.
\end{proof}

\section{XPath Satisfiability}
\label{s:XPath}

In this section, we first describe how XML documents and DTDs
can be represented by data trees and tree automata.
We then introduce a forward fragment of XPath, and a safety subfragment.
By translating XPath queries to
forward alternating tree $1$-register automata,
and applying results from Sections \ref{s:dec.fin} and \ref{s:safety},
we obtain decidability of satisfiablity
for forward XPath on finite documents and
for safety forward XPath on finite or infinite documents.

\subsubsection*{XML Trees}

Suppose $\aalphabet$ is a finite set of element types,
$\aalphabet'$ is a finite set of attribute names,
and $\aalphabet$ and $\aalphabet'$ are disjoint.
An XML document \cite{Bray&Paoli&Sperberg-McQueen98}
is an unranked ordered tree whose every node $\anode$
is labelled by some $\mathtt{type}(\anode) \in \aalphabet$
and by a datum for each element of some
$\mathtt{atts}(\anode) \subseteq \aalphabet'$.
Motivated by processing of XML streams
(cf., e.g., \cite{Olteanu&Furche&Bry04}),
we do not restrict our attention to finite XML documents.

Concerning the data in XML documents, we shall consider only
the equality predicate between data labels.  Equality comparisons
with constants are straightforward to encode using additional attribute names.
Therefore, similarly as \citeN{Bojanczyketal09},
we represent an XML document by
a data tree with alphabet $\aalphabet \cup \aalphabet'$,
where each node $\anode$ is represented by
a sequence of $1 + \length{\mathtt{atts}(\anode)}$ nodes:
the first node is labelled by $\mathtt{type}(\anode)$,
the labels of the following nodes enumerate $\mathtt{atts}(\anode)$,
the children of the last node represent
the first child and the next sibling of $\anode$ (if any),
and for each preceding node in the sequence,
its left-hand child is the next node and its right-hand child is a leaf.
We say that such a data tree is an \emph{XML tree}.

Following \citeN{Benedikt&Fan&Geerts08} and \citeN{Bojanczyketal09},
we assume without loss of generality that document type definitions (DTDs)
\cite{Bray&Paoli&Sperberg-McQueen98} are given as regular tree languages.
More precisely, we consider a DTD to be
a forward nondeterministic tree automaton $\atreeaut$
with alphabet $\aalphabet \cup \aalphabet'$
and without $\emptyword$-transitions.
Such automata can be defined by omitting counters and $\emptyword$-transitions
from ITCA (cf.\ Section~\ref{ss:ITCA}).
Infinite trees are processed in safety mode,
i.e.\ the condition that an infinite run of $\atreeaut$
has to satisfy to be accepting is the same as for finite runs:
for each leaf $\anode$, the state of the configuration at $\anode$ is final.
An XML tree $\adatatree$ as above is regarded to satisfy $\atreeaut$ iff
$\atreeaut$ accepts $\mathtt{tree}(\adatatree)$.

\subsubsection*{Fragments of XPath}

The fragment of XPath \cite{Clark&DeRose99} below
contains all operators commonly found in practice and
was considered in \cite{Benedikt&Fan&Geerts08,Geerts&Fan05}.
The grammars of queries $\aquery$ and qualifiers $\aqual$
are mutually recursive.
The element types $\aletter$ and attribute names $\aletter'$
range over $\aalphabet$ and $\aalphabet'$, respectively.
\begin{eqnarray*}
\aquery & ::= &
\emptyquery \,|\,
{\edown} \,|\, {\eup} \,|\, {\eright} \,|\, {\eleft} \,|\,
{\edown}^* \,|\, {\eup}^* \,|\, {\eright}^* \,|\, {\eleft}^* \,|\,
\aquery / \aquery \,|\, \aquery \cup \aquery \,|\, \aquery[\aqual] \\
\aqual & ::= &
\neg \aqual \,|\, \aqual \wedge \aqual \,|\,
\aquery? \,|\, \aletter \,|\,
\aquery/@\aletter' = \aquery/@\aletter' \,|\,
\aquery/@\aletter' \neq \aquery/@\aletter'
\end{eqnarray*}

We say that a query or qualifier is \emph{forward} iff:
\begin{itemize}
\item
it does not contain $\eup$, $\eleft$, $\eup^*$ or $\eleft^*$;
\item
for every subqualifier of the form
$\aquery_1/@\aletter'_1 \bowtie \aquery_2/@\aletter'_2$,
we have that $\aquery_1 = \emptyquery$
and that $\aquery_2$ is of the form
$\emptyquery$ or ${\edown} / \aquery'_2$ or ${\eright} / \aquery'_2$.
\end{itemize}
A \emph{safety} (resp., \emph{co-safety}) query or qualifier is
one in which each occurence of
$\edown$, ${\edown}^*$ or ${\eright}^*$
is under an odd (resp., even) number of negations.
Since infinite XML documents may contain infinite siblinghoods,
$\edown$, ${\edown}^*$ and ${\eright}^*$ are exactly the queries that
may require existence of a node which can be unboundedly far.

The semantics of queries and qualifiers is standard
(cf., e.g., \cite{Geerts&Fan05}).
We write the satisfaction relations as
$\adatatree, \anode, \anode' \models \aquery$ and
$\adatatree, \anode \models \aqual$, where
$\adatatree$ is an XML tree
$\tuple{\nodes, \aalphabet \cup \aalphabet', \Lambda, \Delta}$,
and $\anode$ and $\anode'$ are $\aalphabet$-labelled nodes.
The definition is recursive over the grammars of queries and qualifiers,
and can be found in Figure~\ref{f:sem.q.q}.
We omit the Boolean cases, and we write
$\widehat{\edown}$, $\widehat{\eup}$, $\widehat{\eright}$ and $\widehat{\eleft}$
for the relations between $\aalphabet$-labelled nodes that correspond to
the child, parent, next-sibling and previous-sibling relations (respectively)
in the document that $\adatatree$ represents.

\begin{figure}
\begin{center}
\begin{eqnarray*}
\adatatree, \anode, \anode' \models \emptyquery
& \stackrel{\mathrm{def}}{\Leftrightarrow} &
\anode = \anode'
\\
\adatatree, \anode, \anode' \models \{\edown, \eup, \eright, \eleft\}
& \stackrel{\mathrm{def}}{\Leftrightarrow} &
\anode \{\widehat{\edown}, \widehat{\eup},
         \widehat{\eright}, \widehat{\eleft}\} \anode'
\\
\adatatree, \anode, \anode' \models \{\edown^*, \eup^*, \eright^*, \eleft^*\}
& \stackrel{\mathrm{def}}{\Leftrightarrow} &
\anode \{\widehat{\edown}^*, \widehat{\eup}^*,
         \widehat{\eright}^*, \widehat{\eleft}^*\} \anode'
\\
\adatatree, \anode, \anode' \models \aquery_1 / \aquery_2
& \stackrel{\mathrm{def}}{\Leftrightarrow} &
\mathrm{there\ exists\ } \anode'' \mathrm{\ such\ that} \\ & &
\adatatree, \anode, \anode'' \models \aquery_1 \mathrm{\ and\ }
\adatatree, \anode'', \anode' \models \aquery_2
\\
\adatatree, \anode, \anode' \models \aquery_1 \cup \aquery_2
& \stackrel{\mathrm{def}}{\Leftrightarrow} &
\adatatree, \anode, \anode' \models \aquery_1 \mathrm{\ or\ }
\adatatree, \anode, \anode' \models \aquery_2
\\
\adatatree, \anode, \anode' \models \aquery[\aqual]
& \stackrel{\mathrm{def}}{\Leftrightarrow} &
\adatatree, \anode, \anode' \models \aquery \mathrm{\ and\ }
\adatatree, \anode' \models \aqual
\\
\adatatree, \anode \models \aquery?
& \stackrel{\mathrm{def}}{\Leftrightarrow} &
\mathrm{there\ exists\ } \anode' \mathrm{\ such\ that\ }
\adatatree, \anode, \anode' \models \aquery
\\
\adatatree, \anode \models \aletter
& \stackrel{\mathrm{def}}{\Leftrightarrow} &
\Lambda(\anode) = \aletter
\\
\adatatree, \anode \models
\aquery_1/@\aletter'_1 \bowtie \aquery_2/@\aletter'_2
& \stackrel{\mathrm{def}}{\Leftrightarrow} &
\mathrm{there\ exist\ }
\anode_1, k_1, \anode_2, k_2
\mathrm{\ such\ that} \\ & &
\adatatree, \anode, \anode_1 \models \aquery_1,
k_1 \leq \length{\mathtt{atts}(\anode_1)},
\Lambda(\anode_1 \cdot 0^{k_1}) = \aletter'_1, \\ & &
\adatatree, \anode, \anode_2 \models \aquery_2,
k_2 \leq \length{\mathtt{atts}(\anode_2)},
\Lambda(\anode_2 \cdot 0^{k_2}) = \aletter'_2, \\ & &
\Delta(\anode_1 \cdot 0^{k_1}) \bowtie \Delta(\anode_2 \cdot 0^{k_2})
\end{eqnarray*}
\end{center}
\caption{Semantics of Queries and Qualifiers}
\label{f:sem.q.q}
\end{figure}

We say that $\adatatree$ satisfies $\aquery$ iff
$\adatatree, \emptyword, \anode' \models \aquery$
for some $\anode'$.

\begin{example}
\label{ex:queries}
Suppose $\aletter'_1, \aletter'_2 \in \aalphabet'$.
The forward query
$\aquery_{\aletter'_1, \aletter'_2} =
 {\eright}^* / {\edown}^*[\emptyquery/@\aletter'_1 =
                          ({\edown} / {\edown}^*)/@\aletter'_2]$
is satisfied by $\aalphabet$-labelled nodes $\anode_0$ and $\anode_1$
iff $\anode_0 \widehat{\eright}^* \widehat{\edown}^* \anode_1$ and
there exists $\anode_2$ such that
$\anode_1 \widehat{\edown}^+ \anode_2$ and
the value of attribute $\aletter'_1$ at $\anode_1$ is equal to
the value of attribute $\aletter'_2$ at $\anode_2$.
Hence, the safety forward query
$\emptyquery[\neg (\aquery_{\aletter'_1, \aletter'_2}?)]$
is satisfied by an XML tree over $\aalphabet$ and $\aalphabet'$
(whose root may have younger siblings) iff
the value of $\aletter'_1$ at a node is never equal to
the value of $\aletter'_2$ at a descendant.
\end{example}

Suppose a query $\aquery$ and a DTD $\atreeaut$
are over the same element types and attribute names.
We say that $\aquery$ is satisfiable relative to $\atreeaut$ iff
there exists an XML tree which satisfies $\aquery$ and $\atreeaut$.
Finitary satisfiability restricts to finite XML trees.

\subsubsection*{Complexity of Satisfiability}

Let us regard a forward qualifier $\aqual$ over
element types $\aalphabet$ and attribute names $\aalphabet'$
as finitely equivalent to an ATRA$_1$ $\aregaut$ with alphabet
$\aalphabet \cup \aalphabet'$
iff, for every finite XML tree $\adatatree$ over $\aalphabet$ and $\aalphabet'$,
and $\aalphabet$-labelled node $\anode$, we have
$\adatatree, \anode \models \aqual$ iff
$\aregaut$ accepts the subtree of $\adatatree$ rooted at $\anode$.
For safety (resp., co-safety) $\aqual$,
safety (resp., co-safety) equivalence is defined by also considering
infinite XML trees and safety (resp., co-safety) acceptance by $\aregaut$.

To formalise the corresponding notions for queries,
we introduce the following kind of automata ``with holes''.
\emph{Query automata} are defined in the same way as
ATRA$_1$ (cf.\ Section~\ref{ss:ATRA}), except that:
\begin{itemize}
\item
transition formulae may contain a new atomic formula $\mathsf{H}$;
\item
no path in the successor graph
from the initial state to a state $\aloc$ such that
$\mathsf{H}$ occurs in some transition formula at $\aloc$
may contain an update edge.
\end{itemize}
The vertices of the successor graph are all states,
there is an edge from $\aloc$ to $\alocbis$ iff
$\alocbis(0, \downarrow)$, $\alocbis(0, \ndownarrow)$,
$\alocbis(1, \downarrow)$ or $\alocbis(1, \ndownarrow)$
occurs in some transition formula at $\aloc$,
and such an edge is called update iff
$\alocbis(0, \downarrow)$ or
$\alocbis(1, \downarrow)$
occurs in some transition formula at $\aloc$.

To define a run of a query automaton
on a data tree $\adatatree$ with the same alphabet
and with respect to a set of nodes $\nodes'$,
we augment the definition of runs of ATRA$_1$
so that whenever a transition formula is evaluated at a node $\anode$,
each occurence of $\mathsf{H}$ is treated as
$\top$ if $\anode \in \nodes'$, and as
$\bot$ if $\anode \notin \nodes'$.
Acceptance of a finite data tree, safety acceptance, and co-safety acceptance,
all with respect to a set of nodes for interpreting $\mathsf{H}$,
are then defined as for ATRA$_1$.

For a query automaton $\aregaut$ and
an ATRA$_1$ or query automaton $\aregaut'$
with the same alphabet and initial states
$\aloc_I$ and $\aloc'_I$ (respectively),
we define the substitution of $\aregaut'$ for the hole in $\aregaut$
by forming a disjoint union of $\aregaut$ and $\aregaut'$,
taking $\aloc_I$ as the initial state,
and substituting each occurence of $\mathsf{H}$ in each transition formula
$\delta(\aloc, \aletter, b)$ of $\aregaut$
by $\delta(\aloc'_I, \aletter, b)$.
Observe that the unreachability in $\aregaut$
of $\mathsf{H}$ from $\aloc_I$ by a path with an update edge
means that the composite automaton transmits
initial register values to $\aregaut'$ without changes.

Now, we say that a forward query $\aquery$ over
element types $\aalphabet$ and attribute names $\aalphabet'$
is finitely equivalent to a query automaton $\aregautbis$ with alphabet
$\aalphabet \cup \aalphabet'$
iff, for every finite XML tree $\adatatree$ over $\aalphabet$ and $\aalphabet'$,
$\aalphabet$-labelled node $\anode$,
and set $\nodes'$ of $\aalphabet$-labelled nodes, we have
$\adatatree, \anode, \anode' \models \aquery$
for some $\anode' \in \nodes'$ iff
$\aregautbis$ accepts the subtree of $\adatatree$ rooted at $\anode$
with respect to $\nodes'$.
For safety (resp., co-safety) $\aquery$,
safety (resp., co-safety) equivalence is defined by also considering
infinite XML trees and safety (resp., co-safety) acceptance by $\aregautbis$.

\begin{theorem}
\label{th:X2RA}
For each forward query $\aquery$ (resp., forward qualifier $\aqual$)
over $\aalphabet$ and $\aalphabet'$, a finitely equivalent
query automaton $\aregautbis_\aquery^{\aalphabet, \aalphabet'}$
(resp., ATRA$_1$ $\aregaut_\aqual^{\aalphabet, \aalphabet'}$)
is computable in logarithmic space.
If $\aquery$ (resp., $\aqual$) is safety,
then it is safety equivalent to
$\aregautbis_\aquery^{\aalphabet, \aalphabet'}$
(resp., $\aregaut_\aqual^{\aalphabet, \aalphabet'}$).
\end{theorem}

\begin{proof}
The translations are defined recursively over
the grammars of queries and qualifiers:
\begin{itemize}
\item
$\aregautbis_\emptyquery^{\aalphabet, \aalphabet'}$,
$\aregautbis_\edown^{\aalphabet, \aalphabet'}$,
$\aregautbis_\eright^{\aalphabet, \aalphabet'}$,
$\aregautbis_{\edown^*}^{\aalphabet, \aalphabet'}$,
$\aregautbis_{\eright^*}^{\aalphabet, \aalphabet'}$ and
$\aregaut_\aletter^{\aalphabet, \aalphabet'}$
are straightforward to define;
\item
$\aregautbis_{\aquery \cup \aquery'}^{\aalphabet, \aalphabet'}$
is formed from $\aregautbis_\aquery^{\aalphabet, \aalphabet'}$
and $\aregautbis_{\aquery'}^{\aalphabet, \aalphabet'}$
by disjunctive disjoint union,
$\aregaut_{\neg \aqual}^{\aalphabet, \aalphabet'}$
is formed from $\aregaut_\aqual^{\aalphabet, \aalphabet'}$
by dualisation, and
$\aregaut_{\aqual \wedge \aqual'}^{\aalphabet, \aalphabet'}$
is formed from $\aregaut_\aqual^{\aalphabet, \aalphabet'}$
and $\aregaut_{\aqual'}^{\aalphabet, \aalphabet'}$
by conjunctive disjoint union
(cf.\ the proof of Proposition~\ref{pr:closure.RA});
\item
to obtain $\aregautbis_{\aquery / \aquery'}^{\aalphabet, \aalphabet'}$,
we substitute $\aregautbis_{\aquery'}^{\aalphabet, \aalphabet'}$
for the hole in $\aregautbis_\aquery^{\aalphabet, \aalphabet'}$;
\item
to obtain $\aregautbis_{\aquery[\aqual]}^{\aalphabet, \aalphabet'}$,
we substitute a conjunctive disjoint union of
$\aregautbis_\emptyquery^{\aalphabet, \aalphabet'}$ and
$\aregaut_\aqual^{\aalphabet, \aalphabet'}$
for the hole in $\aregautbis_\aquery^{\aalphabet, \aalphabet'}$;
\item
$\aregaut_{\aquery?}^{\aalphabet, \aalphabet'}$
is formed from $\aregautbis_\aquery^{\aalphabet, \aalphabet'}$
by substituting $\top$ for $\mathsf{H}$;
\item
an automaton for $\emptyquery/@\aletter'_1 = ({\edown} / \aquery)/@\aletter'_2$
is formed by substituting
the second automaton depicted in Figure~\ref{f:eq.down}
(cf.\ Example~\ref{ex:ATRA} for depicting conventions)
for the hole in $\aregautbis_\aquery^{\aalphabet, \aalphabet'}$,
and substituting the result
for the hole in the first automaton depicted in Figure~\ref{f:eq.down};
\item
the remaining cases in the grammar of qualifiers are handled similarly.
\end{itemize}
The required equivalences, as well as that
if $\aquery$ (resp., $\aqual$) is co-safety
then it is co-safety equivalent to
$\aregautbis_\aquery^{\aalphabet, \aalphabet'}$
(resp., $\aregaut_\aqual^{\aalphabet, \aalphabet'}$),
are shown simultaneously by induction.
\end{proof}

\begin{figure}
\setlength{\unitlength}{2em}
\begin{center}
\begin{picture}(18.5,6)(-.5,1)
\gasset{Nadjust=wh}
\node[Nmarks=i,iangle=180,ilength=1](0)(1,4){$\aloc_0$}
\node(1)(4,4){$\aloc_1$}
\drawedge(0,1){$0, \ndownarrow$}
\node[Nadjustdist=0](1a)(4,7){}
\drawedge[AHnb=0,curvedepth=.5](1,1a)
  {$\aalphabet' \setminus \{\aletter'_1\}$}
\drawedge[curvedepth=.5](1a,1){$0, \ndownarrow$}
\node[Nadjustdist=0](1b)(7,4){}
\node(2)(10,4){$\aloc_2$}
\drawedge[AHnb=0,ELpos=70](1,1b){$\aletter'_1$}
\drawedge[ELpos=30](1b,2){$0, \downarrow$}
\node[Nadjustdist=0](2a)(8.5,7){}
\drawedge[AHnb=0,curvedepth=.5](2,2a){$\aalphabet'$}
\drawedge[curvedepth=.5](2a,2){$0, \ndownarrow$}
\node[Nframe=n](2h)(13,5.5){$\mathsf{H}$}
\drawedge[curvedepth=.5](2,2h){$\aalphabet$}
\node[Nadjustdist=0](2b)(13,4){}
\node(3)(16,4){$\aloc_3$}
\drawedge[AHnb=0,ELpos=70](2,2b){$\aalphabet$}
\drawedge[ELpos=30](2b,3){$0, \ndownarrow$}
\node[Nadjustdist=0](3a)(16,7){}
\drawedge[AHnb=0,curvedepth=.5](3,3a){$\aalphabet'$}
\drawedge[curvedepth=.5](3a,3){$0, \ndownarrow$}
\node[Nadjustdist=0](3b)(13,2.5){}
\drawedge[AHnb=0,curvedepth=.5,ELpos=70](3,3b){$\aalphabet'$}
\drawedge[curvedepth=.5,ELpos=30](3b,2){$1, \ndownarrow$}
\node[Nadjustdist=0](2c)(8.5,1){}
\drawedge[AHnb=0,curvedepth=-.5,ELside=r](2,2c){$\aalphabet$}
\drawedge[curvedepth=-.5,ELside=r](2c,2){$1, \ndownarrow$}
\end{picture}
\end{center}
\begin{center}
\begin{picture}(8,3.5)(-.5,3.5)
\gasset{Nadjust=wh}
\node[Nmarks=i,iangle=180,ilength=1](0p)(1,4){$\aloc'_0$}
\node(1p)(4,4){$\aloc'_1$}
\drawedge(0p,1p){$0, \ndownarrow$}
\node[Nadjustdist=0](1pa)(4,7){}
\drawedge[AHnb=0,curvedepth=.5](1p,1pa){$\aalphabet' \setminus \{\aletter'_2\}$}
\drawedge[curvedepth=.5](1pa,1p){$0, \ndownarrow$}
\node[Nframe=n](1pt)(7,4){$\top$}
\drawedge(1,1pt){$\aletter'_2, =$}
\end{picture}
\end{center}
\caption{Defining
$\aregaut_{\emptyquery/@\aletter'_1 = ({\edown} / \aquery)/@\aletter'_2}
         ^{\aalphabet, \aalphabet'}$}
\label{f:eq.down}
\end{figure}

\begin{theorem}
\label{th:X}
\begin{itemize}
\item[(a)]
For forward XPath and arbitrary DTDs,
satisfiability over finite XML trees is decidable.
\item[(b)]
For safety forward XPath and arbitrary DTDs,
satisfiability over finite or infinite XML trees is decidable.
\end{itemize}
\end{theorem}

\begin{proof}
Given a forward query $\aquery$ and a DTD $\atreeaut$ over
element types $\aalphabet$ and attribute names $\aalphabet'$,
by Theorem~\ref{th:X2RA}, an ATRA$_1$
$\aregaut_{\aquery?}^{\aalphabet, \aalphabet'}$
is computable, which is finitely equivalent to the qualifier $\aquery?$.
We can then compute an ITCA
$\acaut(\aregaut_{\aquery?}^{\aalphabet, \aalphabet'})$
as in the proof of Theorem~\ref{th:ATRA1},
which recognises exactly trees obtained by erasing data from
finite XML trees that satisfy $\aquery$.
To conclude (a), we observe that ITCA
are closed (in logarithmic space) under intersections with
forward nondeterministic tree automata,
and apply Theorem~\ref{th:ITCA}.

For (b), supposing that $\aquery$ is safety,
by Theorem~\ref{th:X2RA} again, an ATRA$_1$
$\aregaut_{\aquery?}^{\aalphabet, \aalphabet'}$
is computable, which is safety equivalent to the qualifier $\aquery?$.
Applying the proof of Theorem~\ref{th:safety} to
$\aregaut_{\aquery?}^{\aalphabet, \aalphabet'}$ and
an ATRA$_1$ whose safety language is empty,
we can compute an ITCANT
$\acaut'(\aregaut_{\aquery?}^{\aalphabet, \aalphabet'})$,
which contains no cycles of $\emptyword$-transitions
and recognises exactly trees obtained by erasing data from
finite or infinite XML trees that satisfy $\aquery$.
It remains to observe that ITCANT with no cycles of $\emptyword$-transitions
are closed (in logarithmic space) under intersections
with forward nondeterministic tree automata,
and to recall that their nonemptiness was shown decidable
also in the proof of Theorem~\ref{th:safety}.
\end{proof}

We remark that, by the proof of \cite[Theorem~5.2]{Demri&Lazic09},
finitary satisfiability for forward XPath with DTDs
is not primitive recursive, even without sibling axes
(i.e., $\eright$ and $\eright^*$).

\section{Concluding Remarks}

It would be interesting to know more about the complexities of
nonemptiness for safety ATRA$_1$ and
satisfiability for safety forward XPath with DTDs.
By Theorem~\ref{th:safety}, the former is decidable and not elementary,
and by Theorem~\ref{th:X}(b), the latter is decidable.

\begin{acks}
We are grateful to the referees for helping us improve the presentation.
\end{acks}

\bibliographystyle{acmtrans}
\bibliography{freeze_j}

\begin{thebibliography}{}

\bibitem[\protect\citeauthoryear{Alpern and Schneider}{Alpern and
  Schneider}{1987}]{Alpern&Schneider87}
{\sc Alpern, B.} {\sc and} {\sc Schneider, F.~B.} 1987.
\newblock Recognizing safety and liveness.
\newblock {\em Distr. Comput.\/}~{\em 2,\/}~3, 117--126.

\bibitem[\protect\citeauthoryear{Benedikt, Fan, and Geerts}{Benedikt
  et~al\mbox{.}}{2008}]{Benedikt&Fan&Geerts08}
{\sc Benedikt, M.}, {\sc Fan, W.}, {\sc and} {\sc Geerts, F.} 2008.
\newblock {XPath} satisfiability in the presence of {DTD}s.
\newblock {\em J. ACM\/}~{\em 55,\/}~2.

\bibitem[\protect\citeauthoryear{Bj\"orklund and Boja\'nczyk}{Bj\"orklund and
  Boja\'nczyk}{2007}]{Bjorklund&Bojanczyk07}
{\sc Bj\"orklund, H.} {\sc and} {\sc Boja\'nczyk, M.} 2007.
\newblock Bounded depth data trees.
\newblock In {\em Automata, Lang. and Program., 34th Int. Coll. (ICALP)}. Lect.
  Notes Comput. Sci., vol. 4596. Springer, 862--874.

\bibitem[\protect\citeauthoryear{Bj{\"o}rklund and Schwentick}{Bj{\"o}rklund
  and Schwentick}{2007}]{Bjorklund&Schwentick07}
{\sc Bj{\"o}rklund, H.} {\sc and} {\sc Schwentick, T.} 2007.
\newblock On notions of regularity for data languages.
\newblock In {\em Fundamentals of Comput. Theory (FCT), 16th Int. Symp.} Lect.
  Notes Comput. Sci., vol. 4639. Springer, 88--99.

\bibitem[\protect\citeauthoryear{Boja\'nczyk, Muscholl, Schwentick, and
  Segoufin}{Boja\'nczyk et~al\mbox{.}}{2009}]{Bojanczyketal09}
{\sc Boja\'nczyk, M.}, {\sc Muscholl, A.}, {\sc Schwentick, T.}, {\sc and} {\sc
  Segoufin, L.} 2009.
\newblock Two-variable logic on data trees and {XML} reasoning.
\newblock {\em J. ACM\/}~{\em 56,\/}~3.

\bibitem[\protect\citeauthoryear{Boja{\'n}czyk, Muscholl, Schwentick, Segoufin,
  and David}{Boja{\'n}czyk et~al\mbox{.}}{2006}]{Bojanczyketal06a}
{\sc Boja{\'n}czyk, M.}, {\sc Muscholl, A.}, {\sc Schwentick, T.}, {\sc
  Segoufin, L.}, {\sc and} {\sc David, C.} 2006.
\newblock Two-variable logic on words with data.
\newblock In {\em 21th IEEE Symp. on Logic in Comput. Sci. (LICS)}. IEEE
  Comput. Soc., 7--16.

\bibitem[\protect\citeauthoryear{Bray, Paoli, and Sperberg-McQueen}{Bray
  et~al\mbox{.}}{1998}]{Bray&Paoli&Sperberg-McQueen98}
{\sc Bray, T.}, {\sc Paoli, J.}, {\sc and} {\sc Sperberg-McQueen, C.} 1998.
\newblock Extensible markup language ({XML}) 1.0.
\newblock W3C Recommendation.

\bibitem[\protect\citeauthoryear{Brzozowski and Leiss}{Brzozowski and
  Leiss}{1980}]{Brzozowski&Leiss80}
{\sc Brzozowski, J.~A.} {\sc and} {\sc Leiss, E.~L.} 1980.
\newblock On equations for regular languages, finite automata, and sequential
  networks.
\newblock {\em Theor. Comput. Sci.\/}~{\em 10,\/}~1, 19--35.

\bibitem[\protect\citeauthoryear{Clark and DeRose}{Clark and
  DeRose}{1999}]{Clark&DeRose99}
{\sc Clark, J.} {\sc and} {\sc DeRose, S.} 1999.
\newblock {XML} path language ({XPath}).
\newblock W3C Recommendation.

\bibitem[\protect\citeauthoryear{David}{David}{2004}]{David04}
{\sc David, C.} 2004.
\newblock Mots et donn{\'e}es infinies.
\newblock M.S.\ thesis, Laboratoire d'Informatique Algorithmique: Fondements et
  Applications, Paris.

\bibitem[\protect\citeauthoryear{deGroote, Guillaume, and Salvati}{deGroote
  et~al\mbox{.}}{2004}]{deGroote&Guillaume&Salvati04}
{\sc deGroote, P.}, {\sc Guillaume, B.}, {\sc and} {\sc Salvati, S.} 2004.
\newblock Vector addition tree automata.
\newblock In {\em 19th IEEE Symp. on Logic in Comput. Sci. (LICS)}. IEEE
  Comput. Soc., 64--73.

\bibitem[\protect\citeauthoryear{Demri and Lazi{\'c}}{Demri and
  Lazi{\'c}}{2009}]{Demri&Lazic09}
{\sc Demri, S.} {\sc and} {\sc Lazi{\'c}, R.} 2009.
\newblock {LTL} with the freeze quantifier and register automata.
\newblock {\em ACM Trans. On Comp. Logic\/}~{\em 10,\/}~3, 30 pp.

\bibitem[\protect\citeauthoryear{Figueira}{Figueira}{2009}]{Figueira09}
{\sc Figueira, D.} 2009.
\newblock Satisfiability of downward {XPath} with data equality tests.
\newblock In {\em 28th ACM SIGACT-SIGMOD-SIGART Symp. on Princ. of Database
  Syst. (PODS)}. ACM, 197--206.

\bibitem[\protect\citeauthoryear{Finkel and Schnoebelen}{Finkel and
  Schnoebelen}{2001}]{Finkel&Schnoebelen01}
{\sc Finkel, A.} {\sc and} {\sc Schnoebelen, P.} 2001.
\newblock Well-structured transitions systems everywhere!
\newblock {\em Theor. Comput. Sci.\/}~{\em 256,\/}~1--2, 63--92.

\bibitem[\protect\citeauthoryear{Geerts and Fan}{Geerts and
  Fan}{2005}]{Geerts&Fan05}
{\sc Geerts, F.} {\sc and} {\sc Fan, W.} 2005.
\newblock Satisfiability of {XPath} queries with sibling axes.
\newblock In {\em Database Program. Lang., 10th Int. Symp. (DBPL)}. Lect. Notes
  Comput. Sci., vol. 3774. Springer, 122--137.

\bibitem[\protect\citeauthoryear{Hall\'e, Villemaire, and Cherkaoui}{Hall\'e
  et~al\mbox{.}}{2006}]{Halle&Villemaire&Cherkaoui06}
{\sc Hall\'e, S.}, {\sc Villemaire, R.}, {\sc and} {\sc Cherkaoui, O.} 2006.
\newblock {CTL} model checking for labelled tree queries.
\newblock In {\em 13th Int. Symp. on Temporal Representation and Reasoning
  (TIME)}. IEEE Comput. Soc., 27--35.

\bibitem[\protect\citeauthoryear{Higman}{Higman}{1952}]{Higman52}
{\sc Higman, G.} 1952.
\newblock Ordering by divisibility in abstract algebras.
\newblock {\em Proc.\ London Math.\ Soc.\ (3)\/}~{\em 2,\/}~7, 326--336.

\bibitem[\protect\citeauthoryear{Jurdzi\'nski and Lazi\'c}{Jurdzi\'nski and
  Lazi\'c}{2007}]{Jurdzinski&Lazic07}
{\sc Jurdzi\'nski, M.} {\sc and} {\sc Lazi\'c, R.} 2007.
\newblock Alternation-free modal mu-calculus for data trees.
\newblock In {\em 22nd IEEE Symp. on Logic in Comput. Sci. (LICS)}. IEEE
  Comput. Soc., 131--140.

\bibitem[\protect\citeauthoryear{Kaminski and Francez}{Kaminski and
  Francez}{1994}]{Kaminski&Francez94}
{\sc Kaminski, M.} {\sc and} {\sc Francez, N.} 1994.
\newblock Finite-memory automata.
\newblock {\em Theor. Comput. Sci.\/}~{\em 134,\/}~2, 329--363.

\bibitem[\protect\citeauthoryear{Kaminski and Tan}{Kaminski and
  Tan}{2008}]{Kaminski&Tan08}
{\sc Kaminski, M.} {\sc and} {\sc Tan, T.} 2008.
\newblock Tree automata over infinite alphabets.
\newblock In {\em Pillars of Comput. Sci., Essays Ded. to Boris (Boaz)
  Trakhtenbrot on the Occ. of His 85th Birthday}. Lect. Notes Comput. Sci.,
  vol. 4800. Springer, 386--423.

\bibitem[\protect\citeauthoryear{Lazi{\'c}}{Lazi{\'c}}{2006}]{Lazic06}
{\sc Lazi{\'c}, R.} 2006.
\newblock Safely freezing {LTL}.
\newblock In {\em FSTTCS: Found. of Softw. Technology and Theor. Comput. Sci.,
  26th Int. Conf.} Lect. Notes Comput. Sci., vol. 4337. Springer, 381--392.
\newblock A revised and extended version is available at
  \texttt{http://arxiv.org/abs/0802.4237}.

\bibitem[\protect\citeauthoryear{L{\"o}ding and Thomas}{L{\"o}ding and
  Thomas}{2000}]{Loding&Thomas00}
{\sc L{\"o}ding, C.} {\sc and} {\sc Thomas, W.} 2000.
\newblock Alternating automata and logics over infinite words.
\newblock In {\em IFIP TCS}. Lect. Notes Comput. Sci., vol. 1878. Springer,
  521--535.

\bibitem[\protect\citeauthoryear{Muller, Saoudi, and Schupp}{Muller
  et~al\mbox{.}}{1986}]{Muller&Saoudi&Schupp86}
{\sc Muller, D.~E.}, {\sc Saoudi, A.}, {\sc and} {\sc Schupp, P.~E.} 1986.
\newblock Alternating automata, the weak monadic theory of the tree, and its
  complexity.
\newblock In {\em Automata, Lang. and Program., 13th Int. Coll. (ICALP)}. Lect.
  Notes Comput. Sci., vol. 226. Springer, 275--283.

\bibitem[\protect\citeauthoryear{Neven, Schwentick, and Vianu}{Neven
  et~al\mbox{.}}{2004}]{Neven&Schwentick&Vianu04}
{\sc Neven, F.}, {\sc Schwentick, T.}, {\sc and} {\sc Vianu, V.} 2004.
\newblock Finite state machines for strings over infinite alphabets.
\newblock {\em ACM Trans. On Comp. Logic\/}~{\em 5,\/}~3, 403--435.

\bibitem[\protect\citeauthoryear{Olteanu, Furche, and Bry}{Olteanu
  et~al\mbox{.}}{2004}]{Olteanu&Furche&Bry04}
{\sc Olteanu, D.}, {\sc Furche, T.}, {\sc and} {\sc Bry, F.} 2004.
\newblock An efficient single-pass query evaluator for {XML} data streams.
\newblock In {\em ACM Symp. on Applied Comput. (SAC)}. ACM, 627--631.

\bibitem[\protect\citeauthoryear{Sakamoto and Ikeda}{Sakamoto and
  Ikeda}{2000}]{Sakamoto&Ikeda00}
{\sc Sakamoto, H.} {\sc and} {\sc Ikeda, D.} 2000.
\newblock Intractability of decision problems for finite-memory automata.
\newblock {\em Theor. Comput. Sci.\/}~{\em 231,\/}~2, 297--308.

\bibitem[\protect\citeauthoryear{Segoufin}{Segoufin}{2006}]{Segoufin06}
{\sc Segoufin, L.} 2006.
\newblock Automata and logics for words and trees over an infinite alphabet.
\newblock In {\em Comput. Sci. Logic (CSL), 20th Int. Works.} Lect. Notes
  Comput. Sci., vol. 4207. Springer, 41--57.

\end{thebibliography}

\begin{received}
  Received May 2008;
  revised March 2010;
  accepted June 2010
  \end{received}

\end{document}